\documentclass[manyauthors]{fundam}

\usepackage{mathptmx}

\usepackage[linesnumbered,ruled,vlined]{algorithm2e}
\usepackage{amssymb}
\usepackage{amsmath}
\usepackage{graphicx}
\usepackage{multirow}
\usepackage{url}

\usepackage{xcolor}
\usepackage{threeparttable}
\usepackage{multirow}

\newcommand{\lpmln}{LP\textsuperscript{MLN}{ }}
\newcommand{\lpmlnend}{LP\textsuperscript{MLN}}

\newcommand{\lglred}[2]{\left(\overline{{#1}_{#2}}\right)^{#2}}

\newcommand{\lse}[1]{HT^m(#1)}
\newcommand{\lsm}[1]{SM^m(#1)}
\newcommand{\asse}[1]{HT^a(#1)}
\newcommand{\assm}[1]{SM^a(#1)}
\newcommand{\kmn}{$k$-$m$-$n$ {}}

\newcommand{\hr}[1]{h(#1)}
\newcommand{\pbr}[1]{pb(#1)}
\newcommand{\nbr}[1]{nb(#1)}

\newcommand{\semise}{\equiv_{s,s}}
\newcommand{\semiseq}{\equiv_{s,s}}
\newcommand{\aspseq}{\equiv_{s,a}}

\newcommand{\srep}[1]{#1^\pm}
\newcommand{\sadd}[1]{#1^+}
\newcommand{\sdel}[1]{#1^-}
\newcommand{\sex}[1]{#1^\oplus}
\newcommand{\srd}[1]{#1^\ominus}

\newcommand{\opsrep}[1]{\Gamma^\pm(#1)}
\newcommand{\opsadd}[1]{\Gamma^+(#1)}
\newcommand{\opsdel}[1]{\Gamma^-(#1)}
\newcommand{\opsex}[1]{\Gamma^\oplus(#1)}
\newcommand{\opsrd}[1]{\Gamma^\ominus(#1)}

\newcommand{\seintadd}[2]{(\sadd{#1}, \sadd{#2})}

\newcommand{\is}[1]{IS(#1)}
\newcommand{\nis}[1]{IS_n(#1)}
\newcommand{\eis}[1]{IS_e(#1)}
\newcommand{\sis}[1]{IS_s(#1)}

\newcommand{\case}[1]{\textbf{Case #1.}{ }}

\newcommand{\setsubeq}[2]{#1 \subseteq #2}

\newcommand{\setdjemp}[2]{#1 \cap #2 = \emptyset}

\newcommand{\expneis}[1]{(#1 \neq \emptyset)}
\newcommand{\expeis}[1]{(#1 = \emptyset)}

\begin{document}

%!TEX root = lpmln_isets.tex

\title{A Syntactic Approach to Studying Strongly Equivalent Logic Programs}

\author{Zhizheng Zhang \thanks{We are grateful to the anonymous referees for their useful comments on the earlier version of this paper. 
		The work was supported by the National Key Research and Development Plan of China (Grant No.2017YFB1002801).} \\
	School of Computer Science and Engineering, \\ Southeast University, China \\
	seu\_zzz{@}seu.edu.cn
	\and Shutao Zhang \corresponding \\ 
	School of Computer Science and Engineering, \\ Southeast University, China \\
	shutao\_zhang{@}seu.edu.cn
	\and Yanghe Feng  \\
	National University of Defense Technology, \\ College of Systems Engineering, Changsha, China \\
	fengyanghe@nudt.edu.cn
	\and Bin Wang \\
	School of Computer Science and Engineering, \\ Southeast University, China \\
	kse.wang{@}seu.edu.cn
}
\address{School of Computer Science and Engineering, Southeast University, No.2 Dongnandaxue Road, Nanjing, Jiangsu Province, China}

% \author[B. Wang et al.]{Bin Wang, Jun Shen, Runqiu Hu, Shutao Zhang, Zhizheng Zhang \\
% School of Computer Science and Engineering, Southeast University, China \\
% \email{\{kse.wang, junshen, hrqnanjing,  shutao\_zhang,  seu\_zzz\}@seu.edu.cn}}

%%%%%%%  parameters to be filled in by copy-editor  %%%%%%%%%%

\setcounter{page}{1}
\publyear{2021}
\papernumber{2060}
\volume{182}
\issue{1}

%%%%%%%%%%%%%%%%%%%%%%%%%%%%%%%%%%%%%%

\maketitle
%\address{Address for correspondence: Z.Zhang, S.Zhang, Department of Computer Science and Enginerring, Southeast University, No.2 Dongnandaxue Road, Nanjing, Jiangsu Province, China; }

\runninghead{Z. Zhang, S. Zhang, Y. Feng, B. Wang}{A Syntactic Approach to Studying Strongly Equivalent Logic Programs}

%!TEX root = lpmln_isets.tex
%!TeX spellcheck = en_US

\begin{abstract}
In the field of Answer Set Programming (ASP), 
two logic programs are strongly equivalent if they are ordinarily equivalent under any extensions. 
This property provides a theoretical foundation for studying many aspects of logic programs such as program simplification and transformation etc.
Therefore, strong equivalence has been investigated extensively for ASP and its extensions such as \lpmlnend. 
In this paper, we present a syntactic approach to studying the strong equivalence of logic programs, 
which provides several interesting results and would help us understand the strong equivalence from a new perspective. 
Firstly, we present the notions of independent sets and five kinds of syntactic transformations (S-* transformations) for logic programs. 
And we investigate the strong equivalence (SE) and non-strong equivalence (NSE) preserving properties of the S-* transformations in the contexts of ASP and \lpmlnend. 
Secondly, based on the properties of the S-* transformations, 
we present a fully automatic algorithm to discover syntactic conditions that preserve strong equivalences (SE-conditions) of ASP and \lpmln programs. 
To discover the SE-conditions efficiently,  we present four kinds of approaches to improve the algorithm. 
Thirdly, we present a preliminary method to simplify the discovered SE-conditions and report the simplified SE-conditions of several kinds of \lpmln  programs. 
After that, we present a discussion on the discovered SE-conditions and some existing problems. 
Finally, we present a comparison between SE-conditions discovering approaches in this paper and in the related work.
\end{abstract}


% \input{todo.tex}

%!TEX root = lpmln_isets.tex

\section{Introduction}

In the field of Answer Set Programming (ASP) \cite{Gelfond1988theSM}, 
the notions of strong equivalences have been extensively investigated for ASP and its extensions, 
since these notions provide a theoretical foundation for studying many aspects of logic programs such as program simplification and transformation \cite{Eiter2004SimplifyingLP,Ji2015SimplifyASP,Cabalar2007Propositional}.
Roughly speaking, two ASP programs $P$ and $Q$ are strongly equivalent, if for any ASP program $R$, 
the extended programs $P \cup R$ and $Q \cup R$ have the same stable models, 
which means the programs $P$ and $Q$ can be replaced each other without considering its context $R$. 
To check the strong equivalence of ASP programs, a model-theoretical approach was presented \cite{Lifschitz2001Strongly,Turner2001SE}, 
i.e. two ASP programs are strongly equivalent iff they have the same models in the logic of Here-and-There (HT-models). 

Besides the model-theoretical approach, several syntactic strong equivalence conditions (which is called SE-conditions for short) are presented to check the strong equivalence of some classes of ASP programs. 
For example, the TAUT and CONTRA conditions guarantee the strong equivalence between an ASP rule and the empty program \cite{Osorio2001Equivalence,Lin2005Discover}.  
The NONMIN, WGPPE, S-HYP, and S-IMP conditions can be used to check the strong equivalence of ASP programs $P$ and $P \cup \{r\}$ \cite{Wang2005SIMP,Wong2008SE}, 
where $r$ is an ASP rule. 
Usually, an SE-condition can only guarantee the strong equivalence of a kind of logic programs. 
To study the SE-conditions of arbitrary ASP programs, 
Lin and Chen \cite{Lin2005Discover} present a computer-aided approach to discovering SE-conditions of \kmn problems for ASP. 
The \kmn problems are to find the SE-conditions of the ASP programs $K \cup M$ and $K \cup N$, 
where $K$, $M$, $N$ are pairwise disjoint ASP programs containing $k$, $m$, and $n$ rules respectively.
Specifically, 
Lin and Chen's approach conjectures a candidate syntactic condition and verifies whether the condition is an SE-condition, 
where part of the verification can be done automatically, and other steps in the approach have to be done manually. 
By the approach, Lin and Chen discover the SE-conditions of several small \kmn problems in ASP. 
An application of the SE-conditions is to simplify logic programs, 
Eiter et al. \cite{Eiter2004SimplifyingLP} present an approach to simplifying ASP program via using existing SE-conditions. 
To study the SE-conditions based simplifications of other logic formalisms, 
we need to study the SE-conditions of these logic formalisms. 
In this paper, we study the SE-conditions of \lpmln programs.

% Therefore, studying the SE-conditions of logic programs is significant in both theoretical and practical aspects. 

% In this paper, we study the SE-conditions of \lpmln programs.
\lpmln \cite{Lee2016Weighted} is a logic formalism that handles probabilistic inference and non-monotonic reasoning by combining ASP and  Markov Logic Networks (MLN) \cite{Richardson2006mln}. 
Recently, Lee and Luo \cite{Lee2019LPMLNSE} and Wang et al. \cite{Wang2019LPMLNSE} have investigated the notions of strong equivalences for \lpmln programs, respectively. 
In their work, several theoretical results including concepts of strong equivalences and corresponding model-theoretical characterizations are presented. 
Among these results, we specially focus on the notion of semi-strong equivalence (i.e. the structural equivalence in Lee and Luo's work), 
due to the notion is a basis for investigating other kinds of strong equivalences of \lpmlnend. 
Similar to the strong equivalence in ASP, semi-strongly equivalent \lpmln programs have the same stable models under any extensions, 
and it can be characterized by HT-models in the sense of \lpmlnend. 
However, the SE-conditions of \lpmln programs have not been investigated systematically. 
A possible way to study the SE-conditions of \lpmln programs is to adapt Lin and Chen's approach to \lpmlnend. 
But in Lin and Chen's approach, the conjecture and part of the verification have to be done manually, 
which makes it nontrivial to find the SE-conditions.

In this paper, we present a novel framework to studying the strong equivalences of logic programs. 
Based on the framework, we present a fully automatic approach to discovering the SE-conditions of logic programs. 
Note that a logic program is either an ASP program or an \lpmln program throughout the paper. 
Our main work is divided into four aspects.

Firstly, we present the notions of independent sets and S-* transformations. 
An independent set of logic programs is a special kind of set of atoms occurred in the programs. 
We show that there is a one-to-one mapping from logic programs and their independent sets, 
which means the programs can be transformed by changing their independent sets. 
Following the idea, 
we present five kinds of syntactic transformations (S-* transformations) 
that can transform logic programs by replacing, adding, and deleting atoms in corresponding independent sets. 
And we investigate the properties of S-* transformations w.r.t. strong equivalences, 
i.e., whether S-* transformations preserve strong equivalences and non-strong equivalences of logic programs, called the SE-preserving and NSE-preserving properties.  

Secondly, we present a fully automatic approach to discovering SE-conditions of logic programs. 
Based on the SE-preserving and NSE-preserving properties of S-* transformations, 
we present the notion of independent set conditions (IS-conditions) and 
show that an IS-condition is an SE-condition if the logic programs constructed from the IS-condition are strongly equivalent. 
According to the results, we present a basic algorithm to discover SE-conditions of \kmn problems through automatically enumerating and verifying their IS-conditions. 
After that, we present four kinds of methods to improve the basic algorithm. 
Experiment results show that these algorithms are efficient to discover SE-conditions of several small \kmn problems in \lpmlnend.

Thirdly, we report the discovered SE-conditions of some \kmn problems of \lpmlnend. 
Since there are too many SE-conditions discovered by the algorithms, 
we present a preliminary method to simplify the SE-conditions and report the simplified SE-conditions. 
Then, we present a discussion w.r.t. the discovered SE-conditions from two aspects. 
Firstly, we report two interesting facts w.r.t. the discovered SE-conditions and discuss potential application and possible theoretical foundation of the facts. 
Secondly, we present some problems on the simplifications of SE-conditions.

Finally, we present a comparison between Lin and Chen's and our approaches, 
which shows the advantage of our approach. 
Our contributions of the paper are two-fold. 
On the one hand, we develop a fully automatic approach to finding SE-conditions for \lpmln and ASP, 
which provides a new example for machine theorem discovering \cite{Lin2018MTD}.
On the other hand, the notions of independent sets and S-* transformations provide a new perspective to understand the strong equivalences of logic programs. 
Especially, two interesting facts reported in this paper imply some unknown theoretical properties of strongly equivalent logic programs.

% Secondly, we discuss the notions of HT-forgetting \cite{Wang2014HTForgetting} and S-* transformations, 
% which shows the potential application of the independent sets and S-* transformations in other problems of logic Programming. 

%!TEX root = lpmln_isets.tex

\section{Preliminaries}
In this section, we review the syntax, semantics, and strong equivalences of ASP and \lpmln programs. 
% Throughout the paper, a logic program is either an ASP program or an \lpmln program.

\label{sec:preliminaries}
% In this section, we briefly review the syntax, semantics, and the strong equivalences of ASP and \lpmlnend. 

\subsection{Syntax}
An ASP program is a finite set of rules of the form 
\begin{equation}
\label{eq:asp-rule-form}
h_1 ~ \vee ~ ... ~\vee ~ h_k ~\leftarrow~ b_{1}, ..., ~b_m, ~not~ c_{1}, ...,~not ~ c_n.
\end{equation}
where $h$s, $b$s, and $c$s are atoms, 
$\vee$ is epistemic disjunction, and $not$ is default negation. 
For an ASP rule $r$, we use $\hr{r}$, $\pbr{r}$, and $\nbr{r}$ to denote the sets of atoms occurred in head, positive body, and negative body of $r$, respectively, 
i.e., $\hr{r} = \{h_1, \dots, h_k\}$, $\pbr{r} = \{b_1, \ldots, b_m\}$, and $\nbr{r} = \{c_1, \ldots, c_n\}$. 
By $at(r)$, we denote the set of atoms occurred in rule $r$, i.e., $at(r) = \hr{r} \cup \pbr{r} \cup \nbr{r}$; 
and by $at(P)$, we denote the set of atoms occurred in a program $P$. 
, i.e., $at(P) = \cup_{r \in P} ~at(r)$.
Based on above notations, a rule $r$ of the form (\ref{eq:asp-rule-form}) can be abbreviated as 
\begin{equation}
\label{eq:asp-rule-form-abbr}
\hr{r} \leftarrow \pbr{r}, ~not ~ \nbr{r}.
\end{equation}
An ASP program is called \textit{ground}, if it contains no variables. 
Usually, a non-ground logic program is considered as a shorthand for the corresponding ground program, 
therefore, we only consider ground logic programs in this paper. 

An \lpmln program is a finite set of weighted ASP rules $w:r$, 
where $r$ is an ASP rule of the form \eqref{eq:asp-rule-form}, and $w$ is a real number denoting the weight of rule $r$.  
% The weight $w$ of an \lpmln rule is either a real number or a symbol ``$\alpha$'' denoting ``infinite weight''. 
For an \lpmln program $P$, by $\overline{P}$, we denote the unweighted ASP counterpart of $P$, 
i.e. $\overline{P} = \{r ~|~ w:r \in P \}$, 
and $P$ is called ground, if its unweighted ASP counterpart $\overline{P}$ is ground. 

\subsection{Semantics}
A ground set $X$ of atoms is called an interpretation in the context of logic programming. 
An interpretation $X$ satisfies an ASP rule $r$, denoted by $X \models r$, 
if $X \cap \hr{r} \neq \emptyset$, $\pbr{r} \not\subseteq X$, or $\nbr{r} \cap X \neq \emptyset$; 
otherwise, $X$ does not satisfy $r$, denoted by $X \not\models r$. 
For an ASP program $P$, interpretation $X$ satisfies $P$, denoted by $X \models P$, if $X$ satisfies all rules of $P$; 
otherwise, $X$ does not satisfy $P$, denoted by $X \not\models P$.
For an interpretation $X$ and an ASP program $P$, the Gelfond-Lifschitz reduct (GL-reduct) $P^X$ is defined as 
\begin{equation}
\label{eq:gl-reduct}
P^X = \{ h(r) \leftarrow \pbr{r}.  ~|~  r \in P \text{ and } \nbr{r} \cap X = \emptyset  \}
\end{equation}
And $X$ is a \textit{stable model}  of $P$, if $X$ satisfies the GL-reduct $P^X$, 
and there does not exist a proper subset $X'$ of $X$ such that $X' \models P^X$. 
By $\assm{P}$, we denote the set of all stable models of an ASP program $P$.

For an \lpmln rule $w:r$ and an interpretation $X$, the satisfiability relation is defined as $X \models w:r$ if $X \models r$ and  $X \not\models w:r$ if $X \not\models r$. 
Similarly, for an \lpmln program $P$, we have $X \models P$ if $X \models \overline{P}$ and $X \not\models P$ if $X \not\models \overline{P}$. 
By $P_X$, we denote the set of rules of an \lpmln program $P$ that can be satisfied by an interpretation $X$, 
called the \textit{\lpmln reduct} of  $P$ w.r.t. $X$, i.e. $P_X=\{w:r \in P ~|~ X \models w:r\}$. 
An interpretation $X$ is a \textit{stable model} of an \lpmln program $P$ if $X$ is a stable model of the ASP program $\overline{P_X}$. 
By $\lsm{P}$, we denote the set of all stable models of an \lpmln program $P$. 
For an \lpmln program $P$ and an interpretation $X$, 
the \textit{weight degree} $W(P,X)$ of $X$ w.r.t. $P$ is defined as $W(P,X) = exp\left(\sum_{w:r \in P_X } w\right)$. 
% Note that the weight $\alpha$ is treated as a symbol in computing the weight degrees. 
\begin{example}
Consider an \lpmln program $P = \{ 1 : a. ~ 1 : a \leftarrow b. ~ 2 : \leftarrow a, not ~c.  \}$ 
and an interpretation $X = \{a\}$. 
It is easy to check that $\overline{P}^X = \{a.  ~ a\leftarrow b. ~ \leftarrow a. \}$ and $X \not\models \overline{P}^X$, 
therefore, we have $X \not\in \assm{\overline{P}}$. 
While under the semantics of \lpmlnend, $\lglred{P}{X} = \{a.  ~ a\leftarrow b.\}$, 
it is clear that $X\in \lsm{P}$ and $W(P, X) = e^2$.
\end{example}

The weight degree is a kind of uncertainty degree of a stable model, there are other kinds of uncertainty degrees in \lpmlnend, which is used in probabilistic inferences. 
In this paper, we focus on the logical aspect of \lpmlnend, therefore, we omit the probabilistic inference part of \lpmln for brevity.

\subsection{Strong Equivalences}
Firstly, we review the definitions of strong equivalences in ASP and \lpmln \cite{Lifschitz2001Strongly,Lee2019LPMLNSE,Wang2019LPMLNSE}. 
% For \lpmln programs, several notions of strong equivalences have been presented, we only introduce two of them.

\begin{definition}[Strong Equivalence for ASP]
\label{def:asp-strong-equivalence}
For ASP programs $P$ and $Q$, they are strongly equivalent, denoted by $P \aspseq Q$, 
if for any ASP program $R$, $\assm{P \cup R} = \assm{Q \cup R}$. 
\end{definition}

\begin{definition}[Strong Equivalences for \lpmlnend]
\label{def:lpmln-strong-equivalences}
For \lpmln programs $P$ and $Q$, we introduce two kinds of notions of strong equivalences:
\begin{itemize}
    \item the programs are semi-strongly equivalent or structural equivalent, denoted by $P \semise Q$, if for any \lpmln program $R$, $\lsm{P \cup R} = \lsm{Q \cup R}$; 
    \item the programs are are w-strongly equivalent, denoted by $P \equiv_{s,w} Q$, if for any \lpmln program $R$, $\lsm{P \cup R} = \lsm{Q \cup R}$, and for each stable model $X \in \lsm{P \cup R}$, $W(P \cup R, X) = W(Q \cup R, X)$. 
\end{itemize}
\end{definition}

For \lpmln programs $P$ and $Q$, it is clear that $P \equiv_{s,w} Q$ implies $P \semise Q$, 
which means the semi-strong equivalence is a foundation of the w-strong equivalence. 
Similarly, other kinds of strong equivalences involving uncertainty degrees are based on the semi-strong equivalence, 
therefore, we first study the syntactic conditions of semi-strongly equivalent \lpmln programs. 
For simplicity, in rest of the paper, we omit the weights of \lpmln rules 
and regard an \lpmln program $P$ as its unweighted ASP counterpart, i.e. $P = \overline{P}$.

Secondly, we review the HT-model based approaches that characterize the strong equivalences of ASP and \lpmlnend. 

% \textbf{Charaterizations.}
\begin{definition}[HT-Interpretation]
\label{def:ht-interpretations}
An HT-interpretation is a pair $(X, Y)$ of interpretations such that $X \subseteq Y$, 
and $(X, Y)$ is called total if $X = Y$, otherwise, it is called non-total.
\end{definition}

\begin{definition}[HT-Models for ASP]
An HT-interpretation $(X, Y)$ satisfies an ASP rule $r$, denoted by $(X, Y) \models r$, 
if $Y \models r$ and $X \models \{r\}^Y$; 
$(X, Y)$ is an HT-model of an ASP program $P$, denoted by $(X, Y) \models P$, if $(X, Y)$ satisfies all rules of $P$. 
By $\asse{P}$, we denote the set of all HT-models of the ASP program $P$.
\end{definition}

\begin{definition}[HT-Models for \lpmlnend]
An HT-interpretation $(X, Y)$ satisfies an \lpmln rule $r$, denoted by $(X, Y) \models r$, 
if $Y \models r$ and $X \models \{r\}_Y^Y$; 
$(X, Y)$ is an HT-model of an \lpmln program $P$, denoted by $(X, Y) \models P$, if $(X, Y)$ satisfies all rules of $P$. 
By $\lse{P}$, we denote the set of all HT-models of the \lpmln program $P$.
\end{definition}

Based on the above definitions, 
Proposition \ref{prop:ht-model-irrelevant-atoms} shows a property of HT-models, 
and Theorem \ref{thm:se-ht-models} shows the characterizations of the strong equivalences reviewed in this section. 

\begin{proposition}
\label{prop:ht-model-irrelevant-atoms}
For a logic program $P$ and an HT-model $(X, Y)$ of $P$, 
if $a$ is an atom such that $a \not\in at(P)$, both of HT-interpretations $(X, Y \cup \{a\})$ and $(X \cup \{a\}, Y \cup \{a\})$ are HT-models of $P$.
\end{proposition}

\begin{theorem}
\label{thm:se-ht-models}
For logic programs $P$ and $Q$, 
\begin{itemize}
    \item $P \aspseq Q$ iff $\asse{P} = \asse{Q}$; 
    \item $P \semiseq Q$ iff $\lse{P} = \lse{Q}$; and 
    \item $P \equiv_{s, w} Q$ iff $\lse{P} = \lse{Q}$ and for any interpretation $X$, $W(P, X) = W(Q, X)$.
\end{itemize}
\end{theorem}

Finally, we show an example of SE-conditions of ASP and \lpmln programs.
For a pair $P$ and $Q$ of logic programs, 
a \textit{syntactic condition} $C$ for the programs is a formula w.r.t. the properties and relationships among the sets of head, positive body, and negative body of each rule in the programs. 
A syntactic condition $C$ is called an \textit{SE-condition}, if for any programs $P$ and $Q$ satisfying the condition $C$, $P$ and $Q$ are strongly equivalent or semi-strongly equivalent;
otherwise, it is called a \textit{non-SE-condition}. 
For example, ASP programs $P = \{r\}$ and $\emptyset$ are strongly equivalent iff the condition in Equation \eqref{eq:asp-010-se-condition} is satisfied, which is called TAUT and CONTRA \cite{Osorio2001Equivalence,Lin2005Discover}. 
\begin{equation}
\label{eq:asp-010-se-condition}
\left( \hr{r} \cup \nbr{r} \right) \cap \pbr{r} \neq \emptyset
\end{equation}
\lpmln programs $P = \{r\}$ and $\emptyset$ are semi-strongly equivalent iff the condition in Equation \eqref{eq:lpmln-010-se-condition} is satisfied \cite{Wang2019LPMLNSE}.
\begin{equation}
\label{eq:lpmln-010-se-condition}
\left( \hr{r} \cup \nbr{r} \right) \cap \pbr{r} \neq \emptyset \text{ or } \hr{r} \subseteq \nbr{r}
\end{equation}
Therefore, Equation \eqref{eq:asp-010-se-condition} and \eqref{eq:lpmln-010-se-condition} are SE-conditions in ASP and \lpmlnend, respectively.  

%!TEX root = lpmln_isets.tex

\section{Independent Sets and S-* Transformations}
In this section, we present a novel approach to studying the strong equivalences of logic programs. 
Firstly, we present the notion of independent set. 
Secondly, we present the notions of  S-* transformations. 
Thirdly, we show the properties of the S-* transformations w.r.t. the  strongly equivalences of logic programs. 
% In this section, we present the notions of independent sets and S-* transformations. 

\subsection{Independent Sets}
For convenient description, a logic program $P$ can be regarded as a tuple of rules, i.e., $P = \langle r_1, ..., r_n \rangle$, 
where $r_i$ is the $i$-th rule of the program $P$. 
Meanwhile, we treat a tuple as an ordered set that may have the same elements, 
therefore, some notations of sets are used for tuples in this paper. 
For example, for a tuple $T$, we use $e_i \in T$ to denote $e_i$ is the $i$-th element of $T$, 
and we use $|T|$ to denote the number of elements of $T$.
For logic programs $P = \langle r_1, ..., r_n \rangle$ and $Q = \langle t_1, ..., t_m \rangle$, 
the pair $\langle P, Q \rangle$ is the concatenation  of $P$ and $Q$, 
i.e.  $\langle P, Q \rangle = \langle r_1, ..., r_n, t_1, ..., t_m\rangle$. 
And a tuple $T = \langle P_1, ..., P_n \rangle$ of logic programs can be recursively defined as 
% $\langle \langle P_1, ..., P_{n-1}\rangle , P_n \rangle$. 
\begin{equation}
    \langle P_1, ..., P_n \rangle = \langle \langle P_1, ..., P_{n-1}\rangle , P_n \rangle
\end{equation}
Based on the above notations, a tuple $T$ of logic programs is regarded as an ordered list of rules. 
For a rule $r$ of the form \eqref{eq:asp-rule-form}, 
it can be represented as a tuple of sets: head $\hr{r}$, positive body $\pbr{r}$, and negative body $\nbr{r}$.
Therefore, a tuple $T$ of logic programs is turned into  
\begin{equation}
\label{eq:program-sets-tuple}
\langle \hr{r_1}, \pbr{r_1}, \nbr{r_1}, \ldots, \hr{r_{|T|}}, \pbr{r_{|T|}}, \nbr{r_{|T|}} \rangle
\end{equation} 
By $TS(T)$, we denote the tuple of sets of the form \eqref{eq:program-sets-tuple} for a tuple $T$ of logic programs. 
It is easy to observe that $|TS(T)| = 3 * |T|$. 
Now, we define the independent sets of a tuple $T$ of logic programs. 

\begin{definition}[Independent Sets]
    \label{def:independent-set}
    Let $T$ be a tuple of logic programs, $N = \{i ~|~ 1 \leq i \leq 3* |T|\}$ a set of positive integers, 
    and $N'$ a non-empty subset of $N$, 
    the independent set $I_{N'}$ w.r.t. $T$ and $N'$ is defined as 
    \begin{equation}
        \label{eq:independent-set}
        I_{N'} =  \bigcap_{i \in N'} S_i - \bigcup_{j \in N - N'} S_j
    \end{equation}
    where $S_i \in TS(T)$ and  $S_j \in TS(T)$. 
    In Equation \eqref{eq:independent-set}, we  call the set $S_i$ ($i \in N'$) an intersection set (i-set) of $I_{N'}$. 
\end{definition}

By Definition \ref{def:independent-set}, 
since there are $2^{3 * |T|} - 1$ non-empty subsets of $N$, 
there are $2^{3 * |T|} - 1$ independent sets for the tuple $T$ of logic programs. 
Intuitively, an independent set $I$ of a tuple $T$ of logic programs contains atoms 
that occur in the i-sets of $I$ and do not occur in the other sets of $TS(T)$. 
It is easy to check that the independent sets w.r.t. $T$ are \textit{pairwise disjoint}, 
and all atoms of $T$ appear in independent sets w.r.t. $T$. 
% , i.e., for any $I_i$ and $I_j$ in $\is{P}$, $I_i \cap I_j = \emptyset$, and $\cup_{I \in \is{P}} = at(P)$.
Therefore, the independent sets can be viewed as a set of fundamental elements to construct logic programs. 
To conveniently distinguish different independent sets w.r.t. $T$, we assign a label to each independent set as follows. 
%TODO: B_{N'} \langle ... \rangle 
%TODO: B_{N'}(m, n)
For an independent set $I_{N'}$ w.r.t. $T$ and a set $N'$ of positive integers, 
let $n = |T|$, 
we use $B(N')$ to denote a tuple of 0s and 1s w.r.t. $I_{N'}$, where the $k$-th element of $B(N')$ is 1 if the $k$-th set $S_k$ of $TS(T)$ is an i-set of $I_{N'}$, otherwise, it is 0,  which is as follows 
\begin{equation}
B(N') = (b_1, \ldots, b_{3 * |T|}), \text{ where } b_k = 
\begin{cases}
1 & \text{if } S_k \text{ is an i-set of } I_{N'}, \text{ i.e., } k \in N'; \\
0 & \text{otherwise}.    
\end{cases}
\end{equation}
Then the tuple $B(N')$ can be viewed as a binary number $(b_1\ldots b_{3n})_2$, 
which can be observed that $1 \leq B(N') < 2^{3 * |T|}$. 
By $B(N', m, n)$ ($1 \leq m < n \leq |B(N')|$), we denote the number w.r.t. the sub-tuple of $B(N')$ from $m$ to $n$, 
i.e., $B(N', m, n) = (b_m, ..., b_n)_2$ where $b_i \in B(N')$ ($m \leq i \leq n$).
Based on the above notations, 
we have established a one-to-one mapping from independent sets to positive integers, 
therefore, we can use $I_k$ ($1 \leq k < 2^{3 * |T|}$) to denote different independent sets w.r.t. a tuple $T$ of logic program, 
where the number $k$ is called the \textit{name} of $I_k$. 
Note that by the definition of independent sets, $I_0$ is not an independent set. 
For a tuple $T$ of logic programs, 
%TODO: fixme < 2^{3 * |T|}
by $\is{T}$, we denote the set of names of all independent sets w.r.t. $T$, i.e. $\is{T} = \{i ~|~ 1 \leq i < 2^{3*|T|}\}$;  
by $\eis{T}$ and $\nis{T}$, we denote the sets of names of all empty and non-empty independent sets in $\is{T}$, respectively. 

\begin{example}
\label{ex:independent-sets}
Consider a logic program $P = \langle r \rangle$, where rule $r$ is 
\begin{equation}
a  \vee b \vee d \leftarrow b, c, ~not~ c.
\end{equation}
we have  $\hr{r} = \{a, b, d\}$, $\pbr{r} = \{b, c\}$, $\nbr{r} = \{c\}$, 
and $TS(P) = \langle \hr{r}, \pbr{r}, \nbr{r} \rangle$. 
It is easy to check that there are seven independent sets w.r.t. $P$, which are as follows 
\begin{itemize}
    \item $I_1 = I_{001} = \nbr{r} - \left( \hr{r} \cup \pbr{r} \right) = \emptyset$, 
    \item $I_2 = I_{010} = \pbr{r} - \left( \hr{r} \cup \nbr{r} \right) = \emptyset$, 
    \item $I_3 = I_{011} = \left(\pbr{r} \cap \nbr{r} \right) - \hr{r} = \{ c\}$, 
    \item $I_4 = I_{100} = \hr{r} - \left( \pbr{r} \cup \nbr{r} \right) = \{a, d \}$, 
    \item $I_5 = I_{101} = \left(\hr{r} \cap \nbr{r} \right) - \pbr{r} = \emptyset$, 
    \item $I_6 = I_{110} = \left(\hr{r} \cap \pbr{r} \right) - \nbr{r} = \{b \}$,  and
    \item $I_7 = I_{111} = \hr{r} \cap \pbr{r} \cap \nbr{r} = \emptyset$. 
    % \item $I_1 = I_2 = I_5 = I_7 = \emptyset$.
\end{itemize}
% Therefore, we have $\eis{P} = \{ I_1, I_2, I_5, I_7 \}$ and $\nis{P} = \{I_3, I_4, I_6\}$.
\end{example}

As shown in Example \ref{ex:independent-sets}, there are 7 independent sets of a rule. 
For a rule $r_k$ in a tuple $T$ of logic programs, 
we use $I_i(r_k)$ ($1 \leq i \leq 7$) to denote independent sets of the rule $r_k$. 
Besides, we define the set $I_0(r_k)$ w.r.t. the rule $r_k$ and the tuple $T$ as 
\begin{equation}
\label{eq:i0-iset}
I_0(r_k) = at(T) - \left( \hr{r_k} \cup \pbr{r_k} \cup \nbr{r_k} \right)
\end{equation}
% we define $I_0(r_k) = at(T) - \left( \hr{r} \cup \pbr{r} \cup \nbr{r} \right)$ for the rule $r_k$. 
where $at(T)$ is the set of atoms occurred in the tuple $T$.
By the law of set difference, ``$A - B$'' is equivalent to ``$A \cap B^c$'', where $B^c$ is the complementary set of $B$. 
Therefore, for an independent set $I_{N'}$ of a tuple $T$ of logic programs, Equation \eqref{eq:independent-set} can be reformulated as 
\begin{equation}
\label{eq:iset-composition}
I_{N'} = 
\bigcap_{1 \leq k \leq |T|} I_{i_k}(r_k), \text{ where } r_k \in T \text{ and }  i_k = B(N', 3k-2, 3k) 
\end{equation}
In Equation \eqref{eq:iset-composition}, we say the independent set $I_{N'}$ is composed by $I_{i_k}(r_k)$, 
denoted by $I_{N'} \sqsubseteq I_{i_k}(r_k)$.
For non-empty independent sets $I$ and $I_i(r)$, it is easy to check that $I \subseteq I_i(r)$ if $I \sqsubseteq I_i(r)$ and $I \cap I_i(r) = \emptyset$ if $I \not\sqsubseteq I_i(r)$.

Based on the above definitions of independent sets, 
given all the independent sets w.r.t. a tuple $T$, we can construct each element $S_i$ of the tuple $TS(T)$ as follows  
\begin{equation}
S_i = \bigcup IS', \text{ where } IS' = \{I ~|~ S_i \text{ is an i-set of } I. \}
\end{equation} 
Therefore, there is a one-to-one mapping from a tuple $T$ of logic programs to the independent sets of $T$. 
Naturally, a tuple $T$ of logic programs can be transformed by adding, deleting, and replacing atoms of its independent sets, 
which is the basic idea of the notions of S-* transformations defined in what follows.

\subsection{S-* Transformations}
By transforming the independent sets of a tuple $T$ of logic programs, we can construct a new tuple of logic programs. 
Here, we present five kinds of ways of transformations, called S-* transformations. 
\begin{definition}[S-* Transformations]
Let $T$ be a tuple of logic programs, $I_k$ an independent set w.r.t. $T$,  and  $a'$ a new atom such that $a' \not\in at(T)$,
the S-* transformations are defined as follows 
\begin{itemize}
    \item single replacement (S-RP) transformation is to replace an atom $a$ in $I_k$ with $a'$, 
    denoted by $\opsrep{I_k, a, a'}$,  where $|I_k| > 0$;
    \item single deletion (S-DL) transformation is to delete an atom $a$ from $I_k$,  denoted by $\opsdel{I_k, a}$, 
    where $a \in I_k$  and $|I_k| > 2$;
    % TODO:  |I_k| \leq 2
    \item single reduce (S-RD) transformation is to delete an atom $a$ from $I_k$, denoted by $\opsrd{I_k, a}$, 
    where  $a \in I_k$ and $0 < |I_k| \leq 2$;
    \item single addition (S-AD) transformation is to add $a'$ to $I_k$, denoted by $\opsadd{I_k, a'}$, 
    where $|I_k| \geq 2$; and 
    \item single extension (S-EX) transformation is to add $a'$ to $I_k$, denoted by $\opsex{I_k, a'}$, 
    where $0 \leq |I_k| < 1$.
\end{itemize}
\end{definition}

For a tuple $T$ of logic programs, we use $\opsrep{T, I_k, a, a'}$ to denote the tuple obtained from $T$ by an S-RP transformation $\opsrep{I_k, a, a'}$, 
and use $\Gamma^\circ(T, I_k, a)$ to denote the tuple obtained from $T$ by other S-* transformations $\Gamma^\circ(I_k, a)$, where $\circ \in \{ -, \ominus, +, \oplus \}$.
For an S-* transformation, we define a series of sub-transformations of S-*, i.e. the S-*-i transformations, 
where an S-*-i transformation can only be used to operate independent set $I$ such that $|I| = i$. 
We use  $\Gamma_i^\circ(\bullet)$ to denote the results obtained by an S-*-i transformation, where $\circ \in \{ \pm, -, \ominus, +, \oplus \}$ and $i \geq 0$.
% Now we take a close look at S-EX and S-RD transformations. 
It is easy to observe that there are only two sub-transformations of S-EX and S-RD, respectively, i.e. S-EX-0, S-EX-1, S-RD-2, and S-RD-1 transformations.

\begin{example}
\label{ex:s-transformations}
Recall the rule $r$ in Example \ref{ex:independent-sets}, 
Table \ref{tab:s-transformation-example} shows different independent sets and rules  obtained  by the S-* transformations, 
where $x$ is a newly introduced atom and $I_k$ ($1 \leq k \leq 7$) are independent sets of the tuple $\langle r \rangle$ in Example \ref{ex:independent-sets}.

\begin{table}
    \centering
    \caption{Rules Obtained from $r$ by S-* Transformations}
    \label{tab:s-transformation-example}
    
    \begin{tabular}{ccc}
    \hline 
    S-*& $\Gamma^\circ(\bullet)$ & $r^\circ$ \\
    \hline
    S-RP & $\srep{I_3} = \opsrep{I_3, c, x} = \{x\}$ & $a \vee b \vee d \leftarrow b, x, not ~x.$ \\
    S-AD & $\sadd{I_4} = \opsadd{I_4, x} = \{a, d, x\}$ & $a \vee b \vee d \vee x \leftarrow b, c, not ~c.$ \\
    S-DL & $(I^+_4)^- = \opsdel{\sadd{I_4}, a}  = \{d, x\}$ & $b \vee d \vee x \leftarrow b, c, not ~c.$ \\
    S-RD-1 & $\srd{I_6} = \Gamma^\ominus_1(I_6, b) = \emptyset$ & $a \vee d \leftarrow c, not ~c.$ \\
    S-RD-2 & $\srd{I_4}= \Gamma^\ominus_2(I_4, d) = \{a\}$ & $a \vee b \leftarrow b, c, not ~c.$ \\
    S-EX-0 & $\sex{I_1} = \Gamma^\oplus_0(I_1, x) = \{x\}$ & $a \vee b \vee d \leftarrow b, c, not ~c, not ~x.$ \\
    S-EX-1 & $\sex{I_3} = \Gamma^\oplus_1(I_3, x) = \{c, x\}$ & $a \vee b \vee d \leftarrow b, c, x, not ~c, not ~x.$ \\
    \hline
    \end{tabular}
    \end{table}
\end{example}

%!TEX root = lpmln_isets.tex

\subsection{Properties of S-* Transformations}
Now we investigate the properties of S-* transformations w.r.t. the strong equivalences of logic programs, 
i.e., the SE-preserving and NSE-preserving properties. 
For brevity, we focus on the investigations for \lpmln programs, and show there are the same results for ASP programs. 
Firstly, we define the notions of SE-preserving and NSE-preserving properties. 
\begin{definition}[SE-preserving and NSE-preserving Properties]
\label{def:se-preserving}
Let the tuple $T = \langle P, Q \rangle$ be a pair of logic programs and $T^\circ = \langle P^\circ, Q^\circ \rangle$ a 
tuple obtained from $T$ by an S-* transformation $\Gamma^\circ(\bullet)$, where $\circ \in \{\pm, +, \oplus, -, \ominus\}$. 
The S-* transformation is called \textit{SE-preserving}, if $P \equiv_{s, \triangle} Q$ implies $P^\circ \equiv_{s, \triangle} Q^\circ$; 
and it is called \textit{NSE-preserving}, if $P \not\equiv_{s, \triangle} Q$ implies $P^\circ \not\equiv_{s, \triangle} Q^\circ$, 
where $\triangle \in \{a, s\}$. 
\end{definition}

\begin{table}
    \centering
    \caption{Properties of S-* Transformations}
    \label{tab:properties-s-*-summary}
    \resizebox{\columnwidth}{!}{
	    \begin{tabular}{cccccccc}
	    \hline
	    & S-RP & S-DL & S-RD-2 & S-RD-1 & S-AD & S-EX-1 & S-EX-0 \\
	    \hline
	    SE-preserving $\left( \text{\lpmlnend} \right)$ & Yes & Yes & Yes & No & Yes & No  & No\\
	    NSE-preserving $\left( \text{\lpmlnend} \right)$ & Yes & Yes & No & No & Yes & Yes & No \\
	    SE-preserving $\left( \text{ASP} \right)$ & Yes & Yes & Yes & No & Yes & No  & No\\
	    NSE-preserving $\left( \text{ASP} \right)$ & Yes & Yes & No & No & Yes & Yes & No \\
	    \hline
	    \end{tabular}
	}
\end{table}

Secondly, we investigate whether an S-* transformation is SE-preserving and NSE-preserving, 
and we focus on the discussion of S-* transformations in the context of \lpmlnend. 
All the main results of this subsection are summarized in Table \ref{tab:properties-s-*-summary}, 
which will be investigated one by one in what follows. 
To investigate the SE-preserving and NSE-preserving properties of the S-* transformations, 
we should investigate the satisfiability between the HT-interpretations and \lpmln programs under S-* transformations by Theorem \ref{thm:se-ht-models}, which is the main idea of the investigations. 
% which is a main idea of the following investigations. 

\textbf{S-RP Transformation.}
For the S-RP transformation, it is clear that there exists a one-to-one mapping from the HT-models of a logic program $P$ to the HT-models of $\srep{P}$, 
where $\srep{P}$ is obtained from $P$ by an S-RP transformation. 
Therefore, we have Theorem \ref{thm:s-rp-se-nse-preserving}.
\begin{theorem}
\label{thm:s-rp-se-nse-preserving}
The S-RP transformation is SE-preserving and NSE-preserving in \lpmlnend. 
\end{theorem}

\textbf{S-DL Transformation in \lpmlnend.}
For the S-DL transformation, consider an \lpmln rule $r$ in a tuple $T$ of logic programs and an independent set $I$ w.r.t. $T$ such that $|I| \geq 3$, 
suppose  rule $\sdel{r}$ is the counterpart of $r$ in $\opsdel{T, I, a}$. 
For an HT-interpretation $(X, Y)$ such that $a \not\in Y$, 
Table \ref{tab:sat-s-dl} shows the satisfiability between $(X, Y)$ and rules  $r$ and $\sdel{r}$ under the S-DL transformation, 
where 
$I_i(r)$ means $I \sqsubseteq I_i(r)$,
``$*$'' means $(X, Y) \models \sdel{r}$ or $(X, Y) \not\models \sdel{r}$,  
``\textemdash'' means corresponding case does not exist, 
$X' = X \cup \{a\}$, and $Y' = Y \cup \{a\}$.
By the results in Table \ref{tab:sat-s-dl}, we investigate the SE-preserving and NSE-preserving properties by Lemma \ref{lem:s-dl-se-preserving} and \ref{lem:s-dl-nse-preserving}, respectively.

\begin{table}
\centering
\caption{Satisfiability between HT-interpretations and Rules under the S-DL Transformation}
\label{tab:sat-s-dl}
\resizebox{\columnwidth}{!}{
\begin{tabular}{ccccccc}
    \hline
    $I \sqsubseteq$   & $(X, Y) \models r$ & $(X, Y) \not\models r$ & $(X, Y') \models r$ & $(X, Y') \not\models r$  & $(X', Y') \models r$ & $(X', Y') \not\models r$  \\
    \hline
    $I_0(r)$ & $(X, Y) \models \sdel{r}$ & $(X, Y) \not\models \sdel{r}$ & $(X, Y) ~\models~  \sdel{r}$ & $(X, Y) ~\not\models~  \sdel{r}$ & $(X, Y) ~\models~  \sdel{r}$ & $(X, Y) ~\not\models~  \sdel{r}$ \\  %\hline

    $I_1(r)$ & $(X, Y) \models \sdel{r}$ & $(X, Y) \not\models \sdel{r}$ & $*$ &  \textemdash & $*$ & \textemdash \\  %\hline

    $I_2(r)$ & * & \textemdash & $*$ & \textemdash & $(X, Y) \models \sdel{r}$ &  $(X, Y) \not\models \sdel{r}$ \\  %\hline

    $I_4(r)$ & $(X, Y) \models \sdel{r}$ & $(X, Y) \not\models \sdel{r}$ & $(X, Y) ~\models~  \sdel{r}$ & $*$ & $*$ & \textemdash \\  %\hline    

    $I_5(r)$ & $(X, Y) \models \sdel{r}$ & $(X, Y) \not\models \sdel{r}$ & $*$ &  \textemdash & $*$ & \textemdash \\ %\hline

    Others & $(X, Y) \models \sdel{r}$ & \textemdash & $(X, Y) \models \sdel{r}$ & \textemdash & $(X, Y) \models \sdel{r}$ & \textemdash  \\ 
    \hline
\end{tabular}
}
\end{table}

\begin{lemma}
\label{lem:s-dl-se-preserving}
The S-DL transformation is SE-preserving in \lpmlnend. 
\end{lemma}

%!TEX root = ../lpmln_isets.tex

% \begin{lemma}[Single Deletion Lemma]
% \label{lem:single-deletion}
% A \kmn tuple $T$ is semi-strongly equivalent iff the \kmn tuple $\sdel{T}$ is semi-strongly equivalent. 
% \end{lemma}

\begin{proof}
Let $T = \langle P, Q \rangle$ be a pair of logic programs such that $P \semise Q$ and set $I$ an independent set w.r.t. $T$ such that $|I| \geq 3$. 
Suppose independent set $\sdel{I} = \opsdel{I, a} = I - \{a\}$ and $\sdel{T} = \langle \sdel{P}, \sdel{Q} \rangle = \opsdel{T, I, a}$. 
For a rule $r$ in $T$, we use $\sdel{r}$ to denote the counterpart of $r$ in $\sdel{T}$.
Since the atom $a$ does not occur in $\sdel{P}$ and $\sdel{Q}$, 
by Proposition \ref{prop:ht-model-irrelevant-atoms}, we only need to investigate the HT-interpretation $(X, Y)$ such that $a \not\in Y$. 
By Table \ref{tab:sat-s-dl}, there are two main kinds of cases. 

\case{1}
If for arbitrary rule $r$ of $T$ such that $I \sqsubseteq I_2(r)$, we have $(X, Y) \models \sdel{r}$. 
By Table \ref{tab:sat-s-dl}, it is easy to check that $(X, Y) \models r$ iff $(X, Y) \models \sdel{r}$ for any rule $r$ of $T$. 
Since the programs $P$ and $Q$ are semi-strongly equivalent, we have $\lse{P} = \lse{Q}$ by Theorem \ref{thm:se-ht-models}, 
which means $\lse{\sdel{P}} = \lse{\sdel{Q}}$. 
Therefore, the programs $\sdel{P}$ and $\sdel{Q}$ are semi-strongly equivalent. 

\case{2}
If there exists a rule $r$ of $T$ such that $I \sqsubseteq I_2(r)$ and $(X, Y) \not\models \sdel{r}$, 
we show $(X, Y)$ is not an HT-model of the programs $\sdel{P}$ and $\sdel{Q}$.
% We use proof by contradiction.
Without loss of generality, suppose $r$ is an \lpmln rule such that $r \in P$.
Since $I \sqsubseteq I_2(r)$ and $|I| \geq 3$, we have 
$I \cap \hr{r} = \emptyset$, $I \subseteq \pbr{r}$, $I \cap \nbr{r} = \emptyset$, and $|\sdel{I}| \neq \emptyset$.
Since $(X, Y) \not\models \sdel{r}$, we have 
$\hr{\sdel{r}} \cap X = \emptyset$, $\hr{\sdel{r}} \cap Y \neq \emptyset$, 
$\pbr{\sdel{r}} \subseteq X \subseteq Y$, 
and $\nbr{\sdel{r}} \cap Y = \emptyset$, 
which means $\sdel{I} \subseteq X \subseteq Y$. 
Let $\sadd{Y} = Y \cup \{a\}$ and $\sadd{X} = X \cup \{a\}$, it is easy to check that $\seintadd{X}{Y} \not\models r$. 
% clear that $I \subseteq \sadd{X} \subseteq \sadd{Y}$, 
% therefore, we have $\seintadd{X}{Y} \not\models r$. 
Since $P$ and $Q$ are semi-strongly equivalent, there must exist a rule $t \in Q$ such that $(\sadd{X}, \sadd{Y}) \not\models t$, 
which means 
$\hr{t} \cap \sadd{Y} \neq \emptyset$, 
$\hr{t} \cap \sadd{X} = \emptyset$, 
$\pbr{t} \subseteq \sadd{X} \subset \sadd{Y}$,  and 
$\nbr{t} \cap \sadd{Y} = \emptyset$. 
Since $I \subseteq \sadd{X} \subseteq \sadd{Y}$, we have $I \cap \hr{t} = \emptyset$ and $I \cap \nbr{t} = \emptyset$. 
Either $I \subseteq \pbr{t}$ or $I \cap \pbr{t} = \emptyset$, we have $\pbr{\sdel{t}} \subseteq X \subseteq Y$, 
therefore, we have $(X, Y) \not\models \sdel{t}$, 
which means $(X, Y)$ is not an HT-model of the programs $\sdel{P}$ and $\sdel{Q}$. 

Combining above results, 
for an HT-interpretation $(X, Y)$ such that $a \not\in Y$ and a rule $r$, 
if $(X, Y) \models r$ and $(X, Y) \not\models \sdel{r}$, 
$(X, Y)$ is not an HT-model of the programs $\sdel{P}$ and $\sdel{Q}$; 
otherwise, we have $(X, Y) \models r$ iff $(X, Y) \models \sdel{r}$. 
Since the programs $P$ and $Q$ are semi-strongly equivalent, it is obvious that programs $\sdel{P}$ and $\sdel{Q}$ are also semi-strongly equivalent, 
Lemma \ref{lem:s-dl-se-preserving} is proven. 
\end{proof}

\begin{lemma}
\label{lem:s-dl-nse-preserving}
The S-DL transformation is NSE-preserving in \lpmlnend. 
\end{lemma}

%!TEX root = ../lpmln_isets.tex

\begin{proof}
Let $T = \langle P, Q \rangle$ be a pair of logic programs such that $P \not\semise Q$ and set $I$ an independent set w.r.t. $T$ such that $|I| \geq 3$. 
Suppose independent set $\sdel{I} = \opsdel{I, a} = I - \{a\}$ and $\sdel{T} = \langle \sdel{P}, \sdel{Q} \rangle = \opsdel{T, I, a}$. 
For a rule $r$ in $T$, we use $\sdel{r}$ to denote the counterpart of $r$ in $\sdel{T}$. 

We use proof by contradiction. 
Assume the programs $\sdel{P}$ and $\sdel{Q}$ are semi-strongly equivalent. 
Without loss of the generality, suppose $(X, Y)$ is an HT-interpretation such that $(X, Y) \models P$ and $(X, Y) \not\models Q$. 
Since $\sdel{P}$ and $\sdel{Q}$ are semi-strongly equivalent, there are two kinds of cases w.r.t. the HT-interpretation $(X, Y)$, i.e. (1) $(X, Y) \not\models \sdel{P}$ and $(X, Y) \not\models \sdel{Q}$, and (2) $(X, Y) \models \sdel{P}$ and $(X, Y) \models \sdel{Q}$. 
According to the relationships among the atom $a$ and interpretations $X$ and $Y$, there are three kinds of cases: (1) $a \not\in Y$, (2) $a \not\in X$ and $a \in Y$, (3) $a \in X$. 
Therefore, the proof is divided into six main cases. 
Since the complete proof is tedious, we only show the proof of a representative case for brevity, the proofs of other cases are similar. 

\case{1}
Suppose $a \not\in Y$, $(X, Y) \not\models \sdel{P}$, and $(X, Y) \not\models \sdel{Q}$,
there is a rule $r \in P$  such that $(X, Y) \models r$ and $(X, Y) \not\models \sdel{r}$. 
By Table \ref{tab:sat-s-dl}, the relationship between $I$ and the independent sets of $r$ is $I \sqsubseteq I_2(r)$, 
which means $\setdjemp{I}{\hr{r}}$, $\setsubeq{I}{\pbr{r}}$, and $\setdjemp{I}{\nbr{r}}$. 
Since $(X, Y) \not\models \sdel{r}$, 
we have $\hr{\sdel{r}} \cap X = \emptyset$, $\hr{\sdel{r}} \cap Y \neq \emptyset$, $\pbr{\sdel{r}} \subseteq X \subseteq Y$, and $\nbr{\sdel{r}} \cap Y = \emptyset$, 
which means $\sdel{I} \subseteq X \subseteq Y$. 
Suppose $b \in I$ and $b \neq a$, let $X' = (X - \{b\}) \cup \{a\}$ and $Y' = (Y - \{b\}) \cup \{a\}$, 
since $|I| \geq 3$, we have $\sdel{I} \cap X' \neq \emptyset$ and $\sdel{I} \not\subseteq X'$.
In rest of the proof, 
we firstly show $(X', Y') \models P$ and $(X', Y') \not\models Q$.
Then, we show $(X', Y') \models \sdel{P}$ and $(X', Y') \not\models \sdel{Q}$. 

Firstly, since $\sdel{I} \subseteq X \subseteq Y$ and $a \not\in Y$, we have $(X, Y) \models t$ for any rule $t$ such that $I \sqsubseteq I_i(t)$ and $i \neq 0$. 
% For such kinds of rules, 
Since $\sdel{I} \cap X' \neq \emptyset$ and $\sdel{I} \not\subseteq X'$, 
it is easy to check that $(X', Y') \models t$ for any rule $t$ such that $I \sqsubseteq I_i(t)$ and $i \neq 0$. 
For the rule $t$ such that $I \sqsubseteq I_0(r)$, since the atoms of $I$ do not appear in $t$, by Proposition \ref{prop:ht-model-irrelevant-atoms}, we have $(X, Y) \models t$ iff $(X', Y') \models t$. 
Combining above results, we have shown that $(X, Y) \models t$ iff $(X', Y') \models t$ for any rule $t$ in the tuple $T$, which means $(X', Y') \models P$ and $(X', Y') \not\models Q$.

Secondly, 
for rule $t$ of $T$ such that $I \sqsubseteq I_0(t)$, 
since $t = \sdel{t}$, we have $(X', Y') \models t$ iff $(X', Y') \models \sdel{t}$. 
For rule $t$ of $T$ such that $I \sqsubseteq I_(t)$ and $i \neq 0$, 
since $\sdel{I} \cap X' \neq \emptyset$ and $\sdel{I} \not\subseteq X'$, 
it is easy to check that $(X', Y') \models \sdel{t}$. 
Above results show that $(X', Y') \models t$ iff $(X', Y') \models \sdel{t}$ for any rule $t$ in the tuple $T$, 
which means $(X', Y') \models \sdel{P}$ and $(X', Y') \not\models \sdel{Q}$.
By Theorem \ref{thm:se-ht-models}, we have $\sdel{P} \not\semise \sdel{Q}$, which contradicts with the assumption. 

\textbf{Others. }
In other cases, it can be shown that either the case does not exist or $\sdel{P} \not\semise \sdel{Q}$, 
therefore, Lemma \ref{lem:s-dl-nse-preserving} is proven. 
\end{proof}

In the proof of Lemma \ref{lem:s-dl-nse-preserving}, 
it can be observed that $|I| \geq 3$ is a critical condition to guarantee the NSE-preserving property of the S-DL transformation, 
which explains why we distinguish two kinds of transformations of deleting atoms. 
Combining Lemma \ref{lem:s-dl-se-preserving} and \ref{lem:s-dl-nse-preserving}, we have shown that the S-DL transformation is SE-preserving and NSE-preserving, 
which is shown in Theorem \ref{thm:s-dl-se-nse-preserving}.
\begin{theorem}
\label{thm:s-dl-se-nse-preserving}
The S-DL transformation is SE-preserving and NSE-preserving in \lpmlnend. 
\end{theorem}

\textbf{S-RD Transformation in \lpmlnend.}
For the S-RD transformation, we discuss the S-RD-1 and S-RD-2 transformations, respectively, 
which is shown in Theorem \ref{thm:s-rd}.

\begin{theorem}
\label{thm:s-rd}
In \lpmlnend, the S-RD-1 transformation is neither SE-preserving nor NSE-preserving; 
and the S-RD-2 transformation is SE-preserving but not NSE-preserving.
\end{theorem}
    
The proof of SE-preserving property of S-RD-2 in Theorem \ref{thm:s-rd} is basically the same as the proof of Lemma \ref{lem:s-dl-se-preserving}. 
Example \ref{ex:s-rd-1-general} shows that the S-RD-1 transformation is neither SE-preserving nor NSE-preserving in general. 
And Example \ref{ex:s-rd-2} shows that the S-RD-2 transformation is not NSE-preserving. 

\begin{example}
\label{ex:s-rd-1-general}
Firstly, consider a tuple of \lpmln rules $T_1 = \langle r \rangle$, 
where $r$ is the rule ``$a \vee c \leftarrow b \vee c.$''. 
By Equation  \eqref{eq:lpmln-010-se-condition}, 
it is easy to check that the rule $r$ is semi-strongly equivalent to the empty program, i.e. $\{r\} \semise \emptyset$. 
For the tuple $T$, the independent set $I_6$ is
\begin{equation}
    I_6 = \hr{r} \cap \pbr{r} - \nbr{r} = \{c\}
\end{equation}
% we have $I_6 = \hr{r} \cap \pbr{r} - \nbr{r} = \{c\}$. 
Let $\srd{T_1} = \langle \srd{r} \rangle = \opsrd{T_1, I_6, c}$, 
we have the rule $\srd{r}$ is ``$a \leftarrow b.$''.
One can check that $\{\srd{r}\} \not\semise \emptyset$, 
since $(\{b\}, \{a, b\}) \not\models \srd{r}$. 
Therefore, the S-RD-1 transformation is not SE-preserving. 

Secondly, consider a tuple of \lpmln rules $T_2 = \langle r_1, r_2 \rangle$, 
where rule $r_1$ is the rule ``$a \vee b. $'' and $r_2$ is the rule ``$b.$''. 
For the HT-interpretation $(X, Y) = (\{a\}, \{a, b\})$, 
it is easy to check that $(X, Y) \models r_1$ and $(X, Y) \not\models r_2$, 
therefore, we have $\{r_1\} \not\semise \{r_2\}$. 
For the tuple $T$, the independent set $I_{32}$ is 
\begin{equation}
I_{32} = I_4(r_1) \cap I_0(r_2) = \hr{r_1} - \left( \pbr{r_1} \cup \nbr{r_1} \cup \hr{r_2} \cup \pbr{r_2} \cup \nbr{r_2} \right) = \{a\}
\end{equation}
Let $\srd{T_2} = \langle \srd{r_1}, \srd{r_2} \rangle = \opsrd{T_2, I_{32}, a}$, 
we have the rule $\srd{r_1}$ is ``$b.$'' and the rule $\srd{r_2}$ is the same as $r_2$. 
It is obvious that $\srd{r_1}$ is semi-strongly equivalent to the rule $\srd{r_2}$. 
Therefore, the S-RD-1 transformation is not NSE-preserving. 
\end{example}

\begin{example}
\label{ex:s-rd-2}
Consider \lpmln programs $P = \langle r_1, r_2 \rangle$ and $Q = \langle r_3, r_4, r_5 \rangle $ 
\begin{equation*}
\begin{array}{lcc}
    P: & a \vee c. & (r_1)\\ & b.  & (r_2) \\ & &
\end{array}
~~~~~~
\begin{array}{lcc}
    Q: & a \vee b \vee c.  & (r_3)\\ & a \vee c \leftarrow b.  & (r_4) \\ & b \leftarrow a, c.  & (r_5)
\end{array}
~~\Rightarrow~~
\begin{array}{lcc}
    \srd{P}: & a. & (\srd{r_1})\\ & b.  & (\srd{r_2}) \\ & &
\end{array}
~~~~~~
\begin{array}{lcc}
    \srd{Q}: & a \vee b.  & (\srd{r_3})\\ & a \leftarrow b.  & (\srd{r_4}) \\ & b \leftarrow a.  & (\srd{r_5})
\end{array}
\end{equation*}
For an HT-interpretation $(X, Y) = (\{a\}, \{a, b\})$, 
it is easy to check that $(X, Y) \not\models P$ and $(X, Y) \models Q$, 
which means the programs $P$ and $Q$ are not semi-strongly equivalent.
For the tuple $T = \langle P, Q \rangle$, there are only two non-empty independent sets $I_{16674}$ and $I_{2324}$ as follows
\begin{eqnarray}
I_{16674} = I_4(r_1) \cap I_0(r_2) \cap I_4(r_3) \cap I_4(r_4) \cap I_2(r_5) =  \{a, c\} \\
I_{2324} = I_0(r_1) \cap I_4(r_2) \cap I_4(r_3) \cap I_2(r_4) \cap I_4(r_5)  = \{b\}
\end{eqnarray}
Let $\srd{T} = \langle \srd{P}, \srd{Q} \rangle = \opsrd{T, I_{16674}, c}$, 
one can check that only the total HT-interpretations are HT-models of the programs $\srd{P}$ and $\srd{Q}$, 
which means the programs $\srd{P}$ and $\srd{Q}$ are semi-strongly equivalent. 
Therefore, the S-RD-2 transformation is not NSE-preserving. 
\end{example}

Above results show that singleton independent sets play a different role from other non-empty independent sets. 
For a tuple $T$ of logic programs, 
we use $\sis{T}$ to denote the set of names of all singleton independent sets of $T$. 

\textbf{S-AD and S-EX Transformations in \lpmlnend.}
For the S-AD and S-EX transformation, it is worth noting that S-AD is the inverse of S-DL and S-EX is the inverse of S-RD, 
therefore, the SE-preserving and NSE-preserving results of S-AD and S-EX can be derived from Theorem \ref{thm:s-dl-se-nse-preserving} and \ref{thm:s-rd} straightforwardly, 
which is shown in Theorem \ref{thm:s-ad-se-nse-preserving} and \ref{thm:s-ex}.

\begin{theorem}
\label{thm:s-ad-se-nse-preserving}
The S-AD transformation is SE-preserving and NSE-preserving in \lpmlnend. 
\end{theorem}

\begin{theorem}
\label{thm:s-ex}
In \lpmlnend, the S-EX-0 transformation is neither SE-preserving nor NSE-preserving; 
and the S-EX-1 transformation is NSE-preserving but not SE-preserving.
\end{theorem}

\textbf{S-* Transformations in ASP.}
Finally, following the same idea of investigating the properties of S-* transformations in \lpmlnend, 
the SE-preserving and NSE-preserving results of S-* transformations in ASP can be obtained. 
Theorem \ref{thm:asp-s-transformations} shows the properties of S-* transformations for ASP programs are the same as the properties in \lpmlnend. 
\begin{theorem}
\label{thm:asp-s-transformations}
The SE-preserving and NSE-preserving results of S-* transformations in ASP are 
\begin{itemize}
    \item the S-RP, S-DL, and S-AD transformations are SE-preserving and NSE-preserving; 
    \item the S-RD-2 transformation is SE-preserving but not NSE-preserving;
    \item the S-EX-1 transformation is NSE-preserving but not SE-preserving; and
    \item the S-RD-1 and the S-EX-0 transformations are neither SE-preserving nor NSE-preserving.
\end{itemize}
\end{theorem}

Now we have shown the SE-preserving and NSE-preserving properties of S-* transformations in ASP and \lpmlnend. 
In next section, we show an application of the properties of S-* transformations, i.e., discovering SE-conditions.

%!TEX root = lpmln_isets.tex

\section{Discovering SE-Conditions for \lpmlnend}
In this section, we present a fully automatic approach to discovering SE-conditions for \lpmln programs, 
which is an application of the independent sets and S-* transformations. 
Firstly, we present the notion of independent set condition (IS-condition) and show the relationships between the notions of SE-condition and IS-condition. 
Secondly, we present a basic algorithm to discover SE-conditions of \kmn problems in \lpmlnend, 
which is a direct application of the properties of the S-* transformations. 
Thirdly, we present four kinds of approaches to improving the basic algorithm by further investigating the properties of S-* transformations. 
Finally, we show the improved algorithm can be used in discovering SE-conditions of several \kmn problems by an experiment.

\subsection{Independent Set Condition}
Here, we define the notion of  independent set conditions (IS-conditions) and show the relationships between the notions of IS-conditions and SE-conditions. 

\begin{definition}[Independent Set Conditions]
\label{def:is-condition}
For a tuple $T$ of logic programs, the independent set condition (IS-condition) $IC(T)$ is a conjunctive formula
\begin{equation}
\label{eq:independent-sets-condition}
IC(T) = \bigwedge_{i \in \nis{T}} (I_i \neq \emptyset) \wedge \bigwedge_{j \in \eis{T}} (I_j = \emptyset) \wedge \bigwedge_{k\in \sis{T}} (|I_k| = 1)
\end{equation}
% where $\wedge$ is logical conjunction.
\end{definition}

\begin{example}
\label{ex:is-condition}
Recall the program $P$ in Example \ref{ex:independent-sets}, 
the IS-condition of $P$ is 
\begin{equation}
IC(P) = \bigwedge_{i \in \{3, 4, 6\}} (I_i \neq \emptyset) \wedge \bigwedge_{j \in \{1, 2, 5, 7\}} (I_j = \emptyset) \wedge \bigwedge_{k\in \{3, 6\}} (|I_k| = 1)
\end{equation}
\end{example}

For a tuple $T$ of logic programs, 
the tuple $T$ is called a singleton tuple if $\sis{T} = \nis{T}$; 
and the IS-condition $IC(T)$ is called a singleton IS-condition if $T$ is a singleton tuple.
For brevity, we associate an IS-condition $IC(T)$ with a tuple $T$ of logic programs, and if $I_k$ is not a singleton independent set in the condition, there are at least two atoms in the independent set $I_k$ of $T$.
For two tuples of programs $T_1 = \langle P_1, ..., P_n\rangle$ and $T_2 = \langle Q_1, ..., Q_n\rangle$, 
we say $T_1$ is structually equal to $T_2$, denoted by $T_1 \approx T_2$, if $|P_i| = |Q_i|$ ($1 \leq i \leq n$).
For tuples $T_1$ and $T_2$, if $T_1 \approx T_2$, we have $\is{T_1} = \is{T_2}$, while the inverse does not hold in general.
Now we define three kinds of relations between IS-conditions. 

\begin{definition}[Relationships between IS-Conditions]
For tuples $T_1$ and $T_2$ of logic programs such that $T_1 \approx T_2$, we have 
\begin{itemize}
    \item $IC(T_1) = IC(T_2)$, if $\nis{T_1} = \nis{T_2}$ and $\sis{T_1} = \sis{T_2}$;
    \item $IC(T_1) < IC(T_2)$, if $\nis{T_1} = \nis{T_2}$ and $\sis{T_2} \subsetneq \sis{T_1}$;  and 
    \item $IC(T_1) \subset IC(T_2)$, if $IC(T_1)$ and $IC(T_2)$ are singleton IS-conditions, and $\nis{T_1} \subsetneq \nis{T_2}$.
\end{itemize}
where $A \subsetneq B$ means $A$ is a proper subset of $B$. 
\end{definition}
Note that $\is{T_1}$ and $\is{T_2}$ are the sets of names of independent set, 
therefore, $\nis{T_1} = \nis{T_2}$ does not mean $T_1$ and $T_2$ have the same independent sets, 
which means $I_k$ is a non-empty independent set of $T_1$ iff $I_k$ is a non-empty independent set of $T_2$ for any $k \in \is{T_1}$.
By Theorem \ref{thm:s-rp-se-nse-preserving} - \ref{thm:asp-s-transformations}, we have following results.

\begin{theorem}
\label{thm:icondition-se-condition}
For tuples of logic programs $T_1 = \langle P_1, Q_1 \rangle$ and $T_2 = \langle P_2, Q_2 \rangle$ 
such that $T_1 \approx T_2$, 
\begin{itemize}
    \item (I) if $IC(T_1) = IC(T_2)$, we have $P_1 \equiv_{s, \triangle} Q_1$ iff $P_2 \equiv_{s, \triangle} Q_2$; and 
    \item (II) if $IC(T_1) < IC(T_2)$, we have $P_1 \not\equiv_{s, \triangle} Q_1$ implies $P_2 \not\equiv_{s, \triangle} Q_2$; and  
    $P_2 \equiv_{s, \triangle} Q_2$ implies $P_1 \equiv_{s, \triangle} Q_1$; 
\end{itemize}
where $\triangle \in \{a, s\}$.
\end{theorem}

Part I of Theorem \ref{thm:icondition-se-condition} can be proven by the SE-preserving and NSE-preserving properties of S-RP, S-DL, and S-AD transformations. 
Since for the tuples $T_1$ and $T_2$ such that $IC(T_1) = IC(T_2)$, 
$T_2$ can be obtained from $T_1$ by using S-RP, S-DL, and S-AD transformations repetitively, 
and all of the transformations are SE-preserving and NSE-preserving. 
Similarly, part II of Theorem \ref{thm:icondition-se-condition} can be proven by properties of S-EX-1 and S-RD-2 transformations.

For a tuple $T = \langle P, Q \rangle$ of logic programs, by Theorem \ref{thm:icondition-se-condition}, 
$IC(T)$ is an SE-conditions for any tuple $T'$ such that $T' \approx T$, if $P  \equiv_{s, \triangle} Q$.  
To verify whether an IS-condition is an SE-condition, 
we construct a tuple of logic programs satisfying the condition firstly. 
Then we can compare the HT-models of the constructed programs. 
Obviously, the verification of an IS-condition can be done automatically. 
An SE-condition $IC(T)$ is the \textit{most general SE-condition}, if there are no tuple $T' = \langle P', Q' \rangle$ such that $IC(T) < IC(T')$ and $P'  \equiv_{s, \triangle} Q'$. 
By $MGIC(T)$, we denote the most general SE-condition w.r.t. a \kmn tuple $T$. 
Part II of Theorem \ref{thm:icondition-se-condition} implies a method to compute most general SE-condition, 
which is shown as follows. 

\begin{corollary}
\label{cor:compute-mg-se-condition}
For a tuple $T = \langle P, Q \rangle$ of \lpmln programs such that $IC(T)$ is an SE-condition, 
there is a tuple $T'$ such that $T' \approx T$ and $IC(T')$ is the most general SE-condition of $T$, 
where $\nis{T'} = \nis{T}$, $\sis{T'}$ is constructed as in Equation \eqref{eq:mg-is-condition},  $a'$ is a newly introduced atoms, and $\triangle \in \{a, s\}$.
\begin{equation}
\label{eq:mg-is-condition}
\sis{T'} = \{k \in \sis{T} ~|~ T' = \langle P', Q' \rangle = \opsex{T, I_k, a'} \text{ and } P' \not\equiv_{s, \triangle} Q' \}
\end{equation}
\end{corollary}

Based on the notions of IS-condition and the most general IS-condition, 
there is a basic algorithm to automatically discover the SE-conditions of \lpmln programs, which is shown in next subsection. 

\subsection{Basic Algorithm to Discover SE-Conditions for \lpmlnend}
Since the IS-conditions can be enumerated and verified automatically, 
Theorem \ref{thm:icondition-se-condition} implies a fully automatic method to discover SE-conditions for \lpmlnend. 
Firstly, we introduce the \kmn problems for \lpmln programs, 
which is presented by Lin and Chen \cite{Lin2005Discover} to study the SE-conditions of ASP programs. 
A \kmn tuple is a triple $ T = \langle K, M, N \rangle$ of \lpmln programs such that $|K| = k$, $|M| = m$, and $|N| = n$. 
% By $T(k, m, n)$, we denote the set of all \kmn tuples.
A \kmn tuple $T$ of \lpmln programs is semi-strongly equivalent, if the \lpmln programs $K \cup M$ and $K \cup N$ are semi-strongly equivalent. 
An SE-condition for a \kmn problem is called a \kmn SE-condition. 
A \kmn SE-condition is called \textit{necessary}, if for any \kmn tuple $T'$ that does not satisfy the condition, $T'$ is not semi-strongly equivalent. 
The \kmn problem is to discover necessary \kmn SE-conditions of \kmn tuples. 
In this paper, it is easy to observe that discovering necessary \kmn SE-conditions is to enumerate and verify all possible IS-conditions of \kmn tuples. 
We use $IC(k, m, n)$ to denote the set of all singleton IS-conditions w.r.t. a \kmn problem. 
Corollary \ref{cor:kmn-necessary-se-condition} shows the form of a necessary \kmn SE-condition, 
which is a direct result of Theorem \ref{thm:icondition-se-condition}.
% \begin{corollary}
% \label{cor:kmn-icondition-se-condition}
% For \kmn tuples $T_1$ and $T_2$ such that $EIC(T_1) = EIC(T_2)$, we have $T_1$ is semi-strongly equivalent iff $T_2$ is semi-strongly equivalent. 
% \end{corollary}

\begin{corollary}
\label{cor:kmn-necessary-se-condition}
The necessary \kmn SE-condition for \lpmln is a formula in disjunctive normal form (DNF)  
\begin{equation}
    \bigvee_{IC(T) \in IC(k,m,n)} MGIC(T) 
\end{equation} 
where $MGIC(T)$ is treated as a falsity if $IC(T)$ is not an \kmn SE-condition. 
% \begin{equation}
% \label{eq:kmn-necessay-SE-condition}
% \bigvee_{T' \in T(k, m, n)} MGIC(T') 
% \end{equation}
\end{corollary}

% Note that there are  infinitely many \kmn tuples of a \kmn problem, therefore, Corollary \ref{cor:kmn-necessary-se-condition} does not mean we need to check all \kmn tuples. 
% In fact, by the definition of IS-condition, we only need to check three kinds of cases for each independent set $I$, 
% i.e. $|I| = 0$, $|I|=1$, and $|I|>1$. 
According to Theorem \ref{thm:icondition-se-condition} and Corollary \ref{cor:compute-mg-se-condition}, 
Algorithm \ref{alg:compute-mg-is-condition} provides a method to verify an IS-condition and compute the most general SE-condition, 
and Algorithm \ref{alg:finding-se-condition} provides a basic method to discover necessary \kmn SE-conditions.

\begin{algorithm}
    \caption{Computing the Most General SE-condition} \label{alg:compute-mg-is-condition}

    \SetKwProg{Fn}{Function}{}{}
    \SetKwFunction{VCMS}{VerifyAndComputeMGSE}
    
    \Fn{\VCMS{$k$, $m$, $n$, $IS_n$, $IS_s$}}{
    	\KwData{the sizes of a \kmn problem: $k$, $m$, $n$;  and \\
    		\qquad \quad the names of non-empty and singleton independent sets: $IS_n$, $IS_s$}
    	\KwResult{most general SE-conditions: $C$}
    	
	    $IS'_s = \emptyset$\; 
	    construct a \kmn tuple $T = \langle K, M, N \rangle$ such that $\nis{T} = IS_n$ and $\sis{T} = IS_s$\;
	    \If{$\lse{K \cup M} = \lse{K \cup N}$ }{
	    	% $IS'_s = \emptyset$\;
	    	\For{$k \in \sis{T}$}{
	    		construct  $\langle K', M', N' \rangle = \opsex{T, I_k, a'}$, where $a' \not\in at(T)$\;
	    		\If{$\lse{K' \cup M'} \neq \lse{K' \cup N'}$}{
	    			$IS'_s = IS'_s \cup \{k\}$\;
	    		}
	    	}
	    	$C = IC(T')$ where $T' \approx T$, $\nis{T'} = \nis{T}$, and $\sis{T'} = IS'_s$ \;
	    	\KwRet{$C$} \;
	    }
	    \Else{
	    	\KwRet{None}\;
    	}
	}
    
\end{algorithm}

\begin{algorithm}
    \caption{Basic Algorithm to Discover \kmn SE-Conditions} \label{alg:finding-se-condition}
    \KwIn{the sizes of a \kmn problem: $k$, $m$, $n$}
    \KwOut{set of the most general SE-conditions: $MGIC$}
    \SetKwFunction{VCMS}{VerifyAndComputeMGSE}
    
    % TODO: fix < 2^{3* (K +m +n)}
    $MGIC= \emptyset$\;
    $IS = \{ k~|~ 1 \leq k < 2^{3 * (k+m+n)}\}$\; 
    % \For{$1 \leq k \leq |IS|$}{
    %     $I_k = \{a_k\}$\;
    % }
    
    \For{$IS_n \subseteq IS$}{
        $C$ = \VCMS{$k, m, n, IS_n, IS_n$}\;
        \If{$C$ is not None}{
            $MGIC = MGIC \cup \{C\}$\;
        }
    }
    return $MGIC$\;
\end{algorithm}

For the 0-1-0 problem, Algorithm \ref{alg:finding-se-condition} discovered 120 SE-conditions, 
which can be viewed as a DNF formula by Corollary \ref{cor:kmn-necessary-se-condition}.  
A DNF formula of the form $(F \wedge a) \vee (F \wedge \neg a)$ can be simplified as $F$, 
therefore, the 120 SE-conditions can be simplified as 
\begin{equation}
\label{eq:010-se-condition-concise}
(I_3 \neq \emptyset) \vee (I_4 = \emptyset) \vee (I_6 \neq \emptyset) \vee (I_7 \neq \emptyset)
\end{equation}
An \lpmln rule $r$ is called \textit{semi-valid} if it satisfies Equation \eqref{eq:010-se-condition-concise}, 
otherwise, it is called \textit{non-semi-valid}. 
The necessary 0-1-0 SE-condition is shown in Equation \eqref{eq:lpmln-010-se-condition}, 
although Equation \eqref{eq:lpmln-010-se-condition} is different from Equation \eqref{eq:010-se-condition-concise} in form, 
these two kinds of conditions are equivalent. 
In addition, since an \lpmln program can be reduced to ASP programs by using choice rules \cite{Lee2019LPMLNSE}, 
the necessary 0-1-0 SE-condition for \lpmln can be proven by necessary 0-1-0 SE-condition for ASP. 
For other \kmn problems in \lpmlnend, 
although the necessary SE-conditions of several \kmn problems in ASP have been presented \cite{Lin2005Discover,Ji2015Discover}, 
it is not trivial to construct the necessary \kmn SE-conditions for \lpmln by the results in ASP. 
Therefore, the searching approach presented in this section is significant. 
% \note{By reductions from \lpmln to ASP, Equation \eqref{eq:010-se-condition-concise} can also be proven in ASP.}

For other \kmn problems, Algorithm \ref{alg:finding-se-condition} is computationally infeasible, 
since the searching space grows extremely rapidly. 
For example, for the 0-1-1 problem, there are $2^{3 * 2} - 1 = 63$ independent sets, 
therefore, Algorithm \ref{alg:finding-se-condition} needs to verify $2^{63}$ singleton IS-conditions, 
which is barely possible to accomplish. 
In what follows, we show how to improve Algorithm \ref{alg:finding-se-condition} by further studying the properties of S-* transformations.

%!TEX root = lpmln_isets.tex

\subsection{Improved Algorithm to Discover SE-Conditions for \lpmlnend}
In this subsection, we present four kinds of methods to improve Algorithm \ref{alg:finding-se-condition}: 
(1) avoiding searching semi-valid rules, 
(2) avoiding searching unnecessary IS-conditions, 
(3) terminating searching process in early stage, 
and (4) optimizing searching spaces. 

\textbf{Avoiding Searching Semi-Valid Rules.} 
A \kmn tuple containing semi-valid rules can be viewed as a small $k'$-$m'$-$n'$ tuple containing no semi-valid rules, 
therefore, there is no need to check corresponding IS-conditions. 
By $NIC(k, m, n)$, we denote the set of all singleton IS-conditions that do not condition semi-valid rules. 
% By $T_n(k,m,n)$, we denote the set of \kmn tuples that do not contain semi-valid rules. 
% For a \kmn problem, Algorithm \ref{alg:finding-se-condition} only needs to verify the singleton IS-condition $IC(T)$ such that $T \in T_n(k, m, n)$. 
Here, we show how to compute  $NIC(k, m, n)$ for a \kmn problem.

For a tuple $T$ of \lpmln programs and a rule $r_i \in T$, an independent set $I$ is called an I-$k$ independent set if $I \sqsubseteq I_k(r_i)$. 
We use $IS(k, m, n)$ to denote the set of names of all independent sets w.r.t. a \kmn problem, 
i.e., $IS(k, m, n) = \{i ~|~ 1 \leq i < 2^{3*(k+m+n)}\}$; 
and use $IS_k(k, m, n)$ to denote the set of names of I-$k$ independent sets of a \kmn problem.
By Equation \eqref{eq:010-se-condition-concise}, 
the sufficient and necessary syntactic condition for non-semi-valid rules is 
\begin{equation}
\label{eq:010-non-se-condtion}
(I_3 = \emptyset) \wedge (I_4 \neq \emptyset) \wedge (I_6 = \emptyset) \wedge (I_7 = \emptyset)
\end{equation}
Therefore, if a \kmn tuple $T$ does not contain semi-valid rules, 
all I-3, I-6, and I-7 independent sets of $T$ should be the empty set, 
i.e., $I_i = \emptyset$ for all $ i \not\in IS'(k,m,n)$, 
where we use $IS'(k, m, n)$ to denote the set of names of independent sets that are not I-3, I-6, and I-7 independent sets, i.e., $IS'(k, m, n) = IS(k, m, n) - (IS_3(k,m,n) \cup IS_6(k,m,n) \cup IS_7(k,m,n))$. 
For a \kmn problem, we use $SIC^1(k, m, n)$ to denote the IS-conditions that contain non-empty I-3, I-6, or I-7 independent sets, 
which is 
\begin{equation}
\label{eq:semi-valid-is-condition-1}
SIC^1(k, m, n) = 
\{ IC(T) \in IC(k, m, n)  ~|~  \nis{T} \not\subseteq IS'(k, m, n)\}
\end{equation}
A \kmn singleton IS-condition $IC(T)$ is called a non-semi-valid IS-condition w.r.t. I-3, I-6, and I-7 independent sets (NSV-IS-condition for short), 
if $IC(T) \not\in SIC^1(k, m, n)$; 
and the tuple $T$ is called a NSV-tuple. 
For example, for the 0-1-1 problem, the independent set $I_{48} = I_6(r_1) \cap I_0(r_2)$, 
therefore, $I_{48}$ is an I-6 independent set. 
Besides, there are other 38 I-3, I-6, and I-7 independent sets for 0-1-1 problem, 
therefore, we only need to check $2^{63 - 39} = 2^{24}$ NSV-IS-conditions.

In addition, by Equation \eqref{eq:010-non-se-condtion}, a \kmn tuple $T$ contains semi-valid rules, 
if there is a rule $r$ in $T$ such that $I \not\sqsubseteq I_4(r)$ for any non-empty independent set $I$ of $T$.
Therefore, such kinds of IS-conditions can be skipped in searching \kmn SE-conditions.
For a \kmn problem, we use $SIC^2(k, m, n)$ to denote the set of all such kinds of IS-conditions, 
which is 
\begin{equation}
\label{eq:semi-valid-is-condition-2}
SIC^2(k, m, n) = 
\{ IC(T) \in IC(k, m, n)  ~|~  \exists r \in T, ~ \forall I \in \nis{T}, ~ I \not\sqsubseteq I_4(r) \}
\end{equation}

Combining above results, for a \kmn problem,  we have
\begin{equation}
\label{eq:is-condition-no-semi-valid}
NIC(k, m, n) = IC(k, m, n) - (SIC^1(k, m, n) \cup SIC^2(k, m, n)) 
\end{equation}

\textbf{Avoiding Verifying Unnecessary IS-Conditions.} 
Algorithm \ref{alg:compute-mg-is-condition} verifies a singleton IS-condition by computing the HT-models of programs constructed from the IS-condition, 
which is unnecessary for some kinds of IS-conditions. 
% A tuple $T$ of logic programs is called an I-$k$ empty tuple, if all I-$k$ independent sets of $T$ are the empty set, 
% and corresponding IS-condition $IC(T)$ is called an I-$k$ empty IS-condition.  
For NSV-tuples of \lpmln programs,  Theorem \ref{thm:s-rd1-ex0-nsv-is-condition} shows that 
the S-EX-0 transformation is NSE-preserving and the S-RD-1 transformation is SE-preserving.

\begin{theorem}
\label{thm:s-rd1-ex0-nsv-is-condition}
For NSV-tuples $T$ and $\sex{T}$ of logic programs such that $T =\langle P, Q \rangle$ and $\sex{T} = \langle \sex{P}, \sex{Q} \rangle = \opsex{T, I_k, a'}$, 
we have $P \not\equiv_{s, \triangle} Q$ implies $\sex{P} \not\equiv_{s, \triangle} \sex{Q}$, 
and $\sex{P} \equiv_{s, \triangle} \sex{Q}$ implies $P \equiv_{s, \triangle} Q$, 
where $\triangle \in \{a, s\}$.
\end{theorem}

The proof of SE-preserving property of the S-RD-1 transformation is basically the same as the proof of Lemma \ref{lem:s-dl-se-preserving}, 
therefore, we omit the proof for brevity. 
For the S-EX-0 transformation, it is the inverse of the S-RD-1 transformation,
therefore, the NSE-preserving property of the S-EX-0 transformation is a direct result of the SE-preserving property of the S-RD-1 transformation. 

Theorem \ref{thm:s-rd1-ex0-nsv-is-condition} can be used to improve Algorithm \ref{alg:finding-se-condition}. 
For an NSV-IS-conditions $IC(T)$, 
if $IC(T)$ is not an SE-condition, 
by Theorem \ref{thm:s-rd1-ex0-nsv-is-condition}, 
arbitrary NSV-IS-condition $IC(T')$ such that $IC(T) \subset IC(T')$ cannot be an SE-condition. 
Therefore, there is no need to verify the HT-models of the programs of $T'$.  
Since verifying HT-models is usually harder in computation, 
Theorem \ref{thm:s-rd1-ex0-nsv-is-condition} can be used to improve the efficiency of Algorithm \ref{alg:finding-se-condition}. 
In particular, if an NSV-IS-condition $(|I_k| = 1) \wedge \bigwedge_{i \neq k} (I_i = \emptyset)$ is not an \kmn SE-condition, 
the independent set $I_k$ should be empty in all \kmn SE-conditions, 
which means the searching space of a \kmn problem can be further reduced. 
For example, 
for the 0-1-1 problem,  there are 8 such kind of independent sets, 
% $I_{44} = I_5(r_1)\cap I_4(r_2)$ is a such kind of independent set. 
% and there are also 7 similar independent sets for 0-1-1 problem. 
therefore, we only need to check $2^{24 - 8} = 2^{16}$ NSV-IS-conditions. 
% which makes the 0-1-1 problem a quite easy searching problem. 

\textbf{Terminating Searching Processes at an Early Stage. }
Algorithm \ref{alg:finding-se-condition} requires to check all IS-conditions in the whole searching space of a \kmn problem, 
which is unnecessary by Theorem \ref{thm:s-rd1-ex0-nsv-is-condition}. 
Theorem \ref{thm:s-rd1-ex0-nsv-is-condition} implies an early terminating criterion for a searching algorithm of the \kmn problems, which is shown in Corollary \ref{cor:is-termination}. 

\begin{corollary}
\label{cor:is-termination}
For a \kmn problem and a constant $k$, if for  all singleton IS-conditions $IC(T) \in NIC(k, m, n)$ such that $|\nis{T}| = k$ are not SE-conditions, 
arbitrary IS-condition $IC(T')  \in NIC(k, m, n)$ such that $|\nis{T'}| > k$ is not an SE-condition. 
\end{corollary}

According to Corollary \ref{cor:is-termination}, the searching space of a \kmn problem can be divided into a series of subspaces, 
where the IS-conditions have the same numbers of non-empty independent sets in each subspace. 
We use $IC_i(k, m, n)$ and $NIC_i(k, m, n)$ to denote the subspaces of the searching space of a \kmn problem, which are
\begin{eqnarray}
& IC_i(k, m, n) = \{ IC(T) \in IC(k, m, n) ~|~  |\nis{T}| = i \} \\
& NIC_i(k, m, n) = IC_i(k, m, n) \cap NIC(k, m, n)    
\end{eqnarray}
A subspace is called a layer of the searching space of a \kmn problem. 
According to Corollary \ref{cor:is-termination}, we can check the IS-conditions of a \kmn problem layer by layer. 
Once all IS-conditions of $NIC_i(k, m, n)$ are not SE-conditions, the searching process can be terminated. 

In addition, by Theorem \ref{thm:s-rd1-ex0-nsv-is-condition}, 
the \kmn NSV-IS-condition $IC(T)$ such that $\sis{T} = IS'(k, m, n)$ should not be an \kmn SE-condition. 
Since if $IC(T)$ is an SE-condition, 
arbitrary \kmn NSV-IS-condition is an SE-condition by the SE-preserving property of the S-RD-1 transformation. 
In other words, if $IC(T)$ is an SE-condition,  arbitrary singleton \kmn NSV-tuple of \lpmln programs is semi-strongly equivalent, 
which is counterintuitive. 

\textbf{Optimizing Searching Spaces.}
In above improvements for Algorithm \ref{alg:finding-se-condition}, 
a searching algorithm needs to check every NSV-IS-conditions for a \kmn problem. 
Here, we show a method to skip searching subspaces for a \kmn problem. 
By Equation \eqref{eq:is-condition-no-semi-valid} and Theorem \ref{thm:s-rd1-ex0-nsv-is-condition}, 
there are two kinds of IS-conditions can be skipped without verifying HT-models, 
i.e., (1) IS-conditions in $SIC^2(k, m, n)$ and (2) IS-conditions related to existing non-SE-conditions. 

For the IS-conditions in $SIC^2(k, m, n)$ of a \kmn problem, 
we can divide the independent sets of a \kmn problem into two subsets, 
i.e. the I-4 and non-I-4 independent sets. 
We use $IS'_4(k,m,n)$ and $\overline{IS'_4}(k,m,n)$ to denote the set of names of I-4 and non-I-4 independent sets in $IS'(k, m, n)$, 
i.e., $IS'_4(k, m, n) = IS'(k, m, n) \cap IS_4(k, m, n)$ and $\overline{IS_4}(k,m,n) = IS'(k, m, n) - IS'_4(k, m, n)$.
To check whether a \kmn IS-condition $IC(T)$ belongs to  $SIC^2(k, m, n)$, 
we only need to check the non-emtpy I-4 independent sets of $T$. 
In other words, we can skip the searching subspaces by only checking the I-4 independent sets of a \kmn problem. 
% For a singleton \kmn tuple $T$, by $T^4$, we denote a singleton \kmn tuple such that $T^4 \approx T$ and $\sis{T^4} = \sis{T} \cap IS_4(k, m, n)$. 
Based on above notations, the subspace $NIC_x(k, m, n)$ of a \kmn problem can be reformulated as 
\begin{equation}
\label{eq:ic-x-no-i4-semi-valid}
\begin{split}
NIC_x(k, m, n) = \{ IC(T) ~|&~ \sis{T} = \sis{T_1} \cup \sis{T_2}, |\sis{T}| = x, \\
& IC(T_1) \in NIC^4_x(k, m, n), \text{ and } \sis{T_2} \subseteq \overline{IS'_4}(k, m, n)\}
\end{split}
\end{equation}
where the set $NIC_x^4(k, m, n)$ of singleton IS-conditions is defined as 
\begin{equation}
\label{eq:nic-4-x}
%NIC^4_x(k, m, n) = 
\{ IC(T) ~|~ 0 < |\sis{T}| \leq x, ~ \sis{T} \subseteq IS'_4(k, m, n), \text{ and } \forall r \in T, ~ \exists i \in \sis{T}, ~ I_i \sqsubseteq I_4(r)\}
\end{equation}
By Equation \eqref{eq:ic-x-no-i4-semi-valid} and \eqref{eq:nic-4-x}, the \kmn searching subspaces such that $IC(T_1) \not\in NIC_x^4(k, m, n)$ can be simply skipped. 
Meanwhile, a layer of the searching space is divided into several small subspaces by Equation \eqref{eq:nic-4-x}, 
which can be searched in parallel. 

Following the same method, the searching space can be further optimized by existing non-SE-conditions. 
Suppose the \kmn singleton IS-condition $IC(T')$ is not an \kmn SE-condition,  
% the independent sets of a \kmn problem can be divided into two subsets w.r.t. $\sis{T'}$, 
% i.e., $\sis{T'}$ and $\overline{IS_s}(T') = IS(k, m, n) - \sis{T'}$.
by Theorem \ref{thm:s-rd1-ex0-nsv-is-condition},  a \kmn searching subspace $NIC_x(k, m, n)$ such that $x > |\sis{T'}|$ 
can be reduced to 
\begin{equation}
\label{eq:nic-x-nse}
\begin{split}
NIC_x(k, m, n) = 
\{IC(T) ~|~ &  \sis{T} = \sis{T_1} \cup  \sis{T_2}, ~|\sis{T}| = x, ~ \\ 
& \sis{T_1} \subsetneq \sis{T'},  \text{ and }
\sis{T_2}  \subseteq IS'(k, m, n) - \sis{T'}  \} 
\end{split}
\end{equation}
For a set of non-SE-conditions, a searching algorithm can repetitively use Equation \eqref{eq:nic-x-nse} for each non-SE-conditions in the set. 
Since the relation $\subset$ between singleton IS-conditions is transitive, 
in above methods, a searching algorithm only need to retain the minimal non-SE-conditions in the sense of the relation $\subset$.

\begin{algorithm}
    \caption{Improved Algorithm to Discover \kmn SE-Conditions} 
    \label{alg:finding-se-condition-improved}
    
    \SetAlgoInsideSkip{}
    \KwIn{the sizes of a \kmn problem: $k$, $m$, $n$}
    \KwOut{set of the most general SE-conditions: $MGIC$, terminating layer: $TR$}
    \SetKwFunction{VCMS}{VerifyAndComputeMGSE}
    
    $MGIC= \emptyset$, $MNSE = \emptyset$ \;
    \tcp{eliminating I-3, I-6, and I-7 independent sets}
    $IS = IS(k, m, n)$, $IS' = IS - (IS_3(k, m, n) \cup IS_6(k, m, n) \cup IS_7(k, m, n))$,  $IS'' = \emptyset$\; 
    
    \For(\tcp*[f]{eliminating searching spaces related to non-SE-conditions}){$x \in IS'$}{ 
    	$C$ = \VCMS{$k, m, n, \{x\}, \{x\}$}\;
        \lIf{$C$ is not None}{
            $IS'' = IS'' \cup \{x\}$
        }
    }

    $IS_4 = IS'' \cap IS_4(k, m, n)$, $\overline{IS_4} = IS'' - IS_4$\;

    \For(\tcp*[f]{searching IS-conditions layer by layer})
        {$0 \leq i \leq |IS''|$}{
        $Space = \emptyset$, $MGIC_i = \emptyset$\;
        \tcp{eliminating searching spaces related to semi-valid rules}
        \For
            {$ISL \subseteq IS_4$ s.t. $|ISL| \leq i$}{
            construct a singleton \kmn tuple $T$ s.t. $\sis{T} = ISL$\;
            \If{ $\forall r \in T$, $\exists I \in \sis{T}$, $I \sqsubseteq I_4(r)$}{
                $Space = Space \cup \{ \langle ISL,~ \overline{IS_4}, ~i - |ISL| \rangle\}$ \label{line:alg3-subspace} \;
            }
        }

        \tcp{eliminating searching spaces related to non-SE-conditions}
        \For
            {$IC(T') \in MNSE$}{
            $Space' = \emptyset$\;
            \For{$sp \in Space$}{
                $IS' = \sis{T'} - sp[1]$\;
                \If{$IS' \subseteq sp[2]$}{
                    \For{$ISL' \subsetneq IS'$ s.t. $|ISL'| \leq sp[3]$}{
                        $Space' = Space' \cup \{ \langle sp[1] \cup ISL', ~ sp[2] - IS', sp[3] - |ISL'| \rangle \}$\;
                    }
                }\lElse{
                    $Space' = Space' \cup \{sp\}$
                }
            }
            $Space = Space'$\;
        }

        \For(\tcp*[f]{verifying IS-conditions by computing HT-models})
            {$sp \in Space$}{
            \For{$ISR \subseteq sp[2]$ s.t. $|ISR| = sp[3]$}{
            	$C$ = \VCMS{$k, m, n, sp[1] \cup ISR, sp[1] \cup ISR$}\;
                \lIf{$C$ is not None}{
                    $MGIC_i = MGIC_i \cup \{ C \}$
                }\lElse{
                    $Max = |\sis{T}|$, $MNSE = MNSE \cup \{ C \}$
                }
            }
        }

        $MGIC = MGIC \cup MGIC_i$\;
        \If(\tcp*[f]{early terminating criterion})
        {$MGIC_i = \emptyset$}{
            % $TR = i$, return $MGIC$, $MNSE$, $TR$, $Max$\;
            $TR = i$, return $MGIC$, $TR$\;
        }
    }
\end{algorithm}

Combining above four kinds of improvements methods, 
an improved method to automatically search \kmn SE-conditions is shown in Algorithm \ref{alg:finding-se-condition-improved}. 
In Line \ref{line:alg3-subspace}, we use a triple $S = \langle ISL,~ \overline{IS_4}, ~i - |ISL| \rangle$ to record a searching subspace. 
The set $IC_S$ of singleton IS-conditions in the subspace $S$ is 
\begin{equation}
IC_S = \{ IC(T) ~|~ ISR \subseteq \overline{IS_4}, ~ |ISR| = i - |ISL|, \text{ and } \sis{T} = ISR \cup ISL\}
\end{equation}
In Line 16 - 26, we use $S[i]$ $(1 \leq i \leq 3)$ to denote the $i$-th element of the subspace triple $S$. 
In what follows, 
we use $MGIC(k, m, n)$ and $MNSE(k, m, n)$ to denote the sets of the \kmn SE-conditions  and the minimal \kmn non-SE-conditions discovered by Algorithm \ref{alg:finding-se-condition-improved}, respectively.

%!TEX root = lpmln_isets.tex

\subsection{Experimental Results}
Now we compare the running times of Algorithm \ref{alg:finding-se-condition} and \ref{alg:finding-se-condition-improved} in solving \kmn problems such that $k + m +n \leq 3$. 
The algorithms were carried out on three servers with Intel Xeon E5-2687W CPU and 100 GB RAM running Ubuntu 16.04\footnote{
An implementation of the searching algorithms and the discovered SE-conditions can be found at \url{https://github.com/wangbiu/lpmln_isets}.
}. 

\begin{table}
    \centering
    \caption{Running Times of Solving \kmn Problems}
    \label{tab:kmn-running-times}
    \resizebox{\columnwidth}{!}{%
    \begin{tabular}{cccccccccc}
        \hline
        & $|IS|$ & $|IS'|$ & $|IS''|$ & $TR$ & $|MGIC|$ & $|MNSE|$ & $Max$ & Alg. \ref{alg:finding-se-condition} & Alg. \ref{alg:finding-se-condition-improved} \\
        \hline
        0-1-0 & 7 & n/a & n/a & 7 & 120 & 1 & 1 & $<$ 1 s & n/a \\
        0-1-1 & 63 & 24 & 16 & 7 & 32 & 18 & 2 & 3 h 13 m & 21 s \\
        1-1-0 & 63 & 24 & 20 & 12 & 1024 & 13 & 2 & 8 h 2 m & 54 s \\
        0-2-1 & 511 & 63 & 33 & 7 & 60 & 71 & 3 & n/a  & 35 s\\
        1-2-0 & 511 & 63 & 42 & 15 & 10240 & 81 & 3 & n/a & 15 m 32 s \\
        1-1-1 & 511 & 63 & 45 & 16 & 39392 & 409 & 3 & n/a & 39 m 16 s \\
        2-1-0 & 511 & 63 & 54 & $\leq$ 28 & $\leq$ 249913344 & $\geq$ 984 & $\geq$ 3 & n/a & 25 h 37 m \\
        \hline
    \end{tabular}
	}
    \end{table}

Table \ref{tab:kmn-running-times} shows the running records of Algorithm \ref{alg:finding-se-condition} and \ref{alg:finding-se-condition-improved}, 
where the columns Alg. \ref{alg:finding-se-condition} and Alg. \ref{alg:finding-se-condition-improved} show the running times of solving a \kmn problem via using Algorithm \ref{alg:finding-se-condition} and Algorithm \ref{alg:finding-se-condition-improved} respectively, 
other columns show some important data in Algorithm \ref{alg:finding-se-condition-improved}, 
which are as follows
\begin{itemize}
    \item $|IS|$ (Line 2) is the number of the independent sets w.r.t. a \kmn problem; 
    \item $|IS'|$ (Line 2)  is the  number of non-I-3, non-I-4, and non-I7 independent sets of $IS$;
    \item $|IS''|$ (Line 6)  is the number of independent sets of $IS'$ that can construct \kmn SE-conditions;
    \item $TR$ (Line 33)  is the terminating layer of a \kmn problem, i.e. for any \kmn IS-condition $IC(T)$ such that $|\sis{T}| = TR$, $IC(T)$ is not a \kmn SE-condition; 
    \item $|MGIC|$ (Line 31) is the number of the most general SE-conditions of a \kmn problem; 
    \item $|MNSE|$ (Line 30) is the number of the minimal non-SE-conditions of a \kmn problem;  
    \item $Max$ (Line 30) is the maximal number of singleton independent sets w.r.t. a minimal non-SE-condition, 
    i.e. $Max = |\sis{T}|$, where $IC(T) \in MNSE$ and for any $IC(T') \in MNSE$, $|\sis{T'}| \leq |\sis{T}|$.
\end{itemize}
Note that for the \kmn problems such that $k + m + n > 2$, the set $IS'$ does not contain I-5 independent sets, 
which is an application of the 0-1-1 SE-conditions. 
For the \kmn problems such that $k + m + n > 1$, 
their searching spaces in Algorithm \ref{alg:finding-se-condition} are constructed from $IS'$, 
since original Algorithm \ref{alg:finding-se-condition} can only be used to solve the 0-1-0 problem. 

For the 2-1-0 problem, there are too many IS-conditions that need to be verified. 
And for the IS-conditions consisting of more than 20 non-empty independent sets, 
it usually takes a very long time to verify them by using Algorithm \ref{alg:compute-mg-is-condition}. 
Therefore, for 2-1-0 singleton IS-condition $IC(T)$ such that $|\sis{T}| > 3$, 
we skip the verification, 
which means there are at most 249913344 2-1-0 SE-conditions and at least 984 minimal 2-1-0 non-SE-conditions. 
Except for the 2-1-0 problem, the other \kmn problems in Table \ref{tab:kmn-running-times} are called  verified \kmn problems. 

For the verified \kmn problems, 
Algorithm \ref{alg:finding-se-condition-improved} can be done in a reasonable amount of time, 
and corresponding SE-conditions are reported in next section.

%!TEX root = lpmln_isets.tex

\section{SE-conditions of Verified \kmn Problems}
In this section, we report the discovered SE-conditions of the verified \kmn problems. 
Firstly, we present a preliminary approach to simplifying the SE-conditions. 
Secondly, we report the simplified SE-conditions of the verified \kmn problems. 
% Thirdly, we show an application of the \kmn  SE-conditions. 
Finally, we present a discussion on the discovered \kmn SE-conditions.

\subsection{Simplification of \kmn SE-conditions}
As shown in Table \ref{tab:kmn-running-times}, there are many SE-conditions of a \kmn problem, 
which is not convenient in practical use. 
For the 0-1-0 problem, 
Equation \eqref{eq:010-se-condition-concise} shows that the discovered 120 SE-conditions can be simplified, 
which is simplified by a symbolic computing toolkit such as SymPy \cite{Meurer2017Sympy}. 
But for other \kmn problems, there are too many symbols in SE-conditions, 
which makes SymPy not applicable. 
Therefore, we need to study how to simplify SE-conditions of a \kmn problem.

For a set $IC$ of \kmn SE-conditions, by $CIS_e(IC)$ and $CIS_n(IC)$, we denote the set of names of common empty and non-empty independent sets of the conditions in $IC$, respectively, which are 
\begin{equation}
CIS_\triangle(IC) = \bigcap_{IC(T) \in IC} IS_\triangle(T), \text{ where } \triangle \in \{e, n\}
\end{equation}
The notion of clique for a set of \kmn SE-conditions is defined as follows.

\begin{definition}[Clique]
A set $IC$ of \kmn SE-conditions is called a clique, if both of the following conditions are satisfied
\begin{itemize}
    \item there exists an SE-condition $IC(T)$ of $IC$ such that for any other SE-condition $IC(T')$ of $IC$, 
    $\sis{T'} =  \nis{T'} \cap \sis{T}$ and $\nis{T'} \subsetneq \nis{T}$; 
    \item $|IC| = 2 ^ x$, where $x = |IS(k, m, n)| - |CIS_e(IC)| - |CIS_n(IC)|$. 
\end{itemize}
And the SE-condition $IC(T)$ is called the max-SE-condition of $IC$. 
\end{definition}
For a clique $IC$ of \kmn SE-conditions and its max-SE-condition $IC(T)$, 
it is easy to check that the clique $IC$ can be simplified as 
\begin{equation}
%	Sim(IC) = \bigwedge_{i \in CIS_n(IC)} \expneis{I_i} \wedge \bigwedge_{j \in CIS_e(IC)} \expeis{I_j} \wedge 
%	\bigwedge_{k \in \sis{T}} (|I_k| \leq 1)
    Sim(IC) = \bigwedge_{i \in CIS_n(IC)} \expneis{I_i} \wedge \bigwedge_{j \in CIS_e(IC)} \expeis{I_j} \wedge 
    \bigwedge_{k \in \sis{T}} (|I_k| \leq 1)
\end{equation}
A clique $IC$ of \kmn SE-conditions is called a maximal clique of $MGIC(k, m, n)$, 
if $IC \subseteq MGIC(k, m, n)$ and there does not exist a clique such that $IC' \subseteq MGIC(k, m, n)$  and  $IC \subsetneq IC'$. 
Based on the above notions, simplifying \kmn SE-conditions of $MGIC(k, m, n)$ is turned into finding the maximal cliques of $MGIC(k, m, n)$. 
For a \kmn problem, by $MC(k, m, n)$, we denote the set of maximal cliques of $MGIC(k, m, n)$, 
the simplified \kmn SE-condition is 
\begin{equation}
\bigvee_{IC \in MC(k, m, n)} Sim(IC)
\end{equation}

\begin{example}
Consider following SE-conditions 
\begin{align*}
    \expneis{I_1} \wedge \expneis{I_2} \wedge \expneis{I_3}  \tag{IC1} \\
    \expneis{I_1} \wedge \expeis{I_2} \wedge \expneis{I_3}  \tag{IC2} \\
    \expneis{I_1} \wedge \expneis{I_2} \wedge \expeis{I_3}  \tag{IC3} 
\end{align*} 
It is easy to check that there are two maximal cliques in the above three SE-conditions, 
i.e. $MC_1 = \{IC1, ~IC2\}$ and $MC_2 = \{IC1, ~IC3\}$. 
Therefore, the simplified SE-condition is 
\begin{equation}
    \expneis{I_1} \wedge \left(\expneis{I_2} \vee  \expneis{I_3}\right)
\end{equation}
\end{example}

Usually, some \kmn SE-conditions in $MGIC(k, m, n)$ consist of singleton independent sets, 
by $IC_s(k, m, n)$, we denote the set of names of singleton independent sets of all \kmn  SE-conditions, which is 
\begin{equation}
    IC_s(k, m, n) = \bigcup_{IC(T) \in MGIC(k, m, n)} \sis{T}
\end{equation}
A subset $IC$ of $MGIC(k, m, n)$ is called singleton independent set irrelevant (SIS-irrelevant for short), 
if one of following conditions is satisfied
\begin{itemize}
    \item for any SE-condition $IC(T) \in IC$, $\sis{T} = \emptyset$; or
    \item for any SE-condition $IC(T) \in IC$, $\nis{T} \cap IC_s(k, m, n) = \sis{T}$.
\end{itemize}
It is easy to check that only an SIS-irrelevant subset of $MGIC(k, m, n)$ could be a clique. 
Therefore, to simplify the \kmn SE-conditions, we divide the set $MGIC(k, m, n)$ into several maximal SIS-irrelevant subsets. 
For each maximal SIS-irrelevant subset of $MGIC(k, m, n)$, Algorithm \ref{alg:finding-max-cliques} provides a preliminary method to find maximal cliques of the subset. 

\begin{algorithm}
\caption{Finding Maximal Cliques} \label{alg:finding-max-cliques}
\KwIn{An SIS-irrelevant subset of $MGIC(k, m, n)$: $IC$}
\KwOut{Maximal Cliques of $IC$: $MaxCliques$}
$MaxIC = \{IC(T) \in IC ~|~ \forall IC(T') \in IC, ~ \nis{T} \not\subseteq \nis{T'}\}$, $Cliques = \emptyset$\;
\For{$IC(T) \in MaxIC$}{
    $IC' = \{ IC(T') \in IC ~|~ \nis{T'} \subseteq \nis{T} \}$, $NIS = \bigcup_{IC(T') \in IC'} \nis{T'}$, $IC''= \emptyset$\; 
    \For{$NIS' \subsetneq NIS$}{
        $CQ_1 = \{ IC(T') \in IC' ~|~  NIS' \subseteq \nis{T'}\}$, 
        $CQ_2 = \{ IC(T') \in IC' ~|~  NIS' \subseteq \eis{T'}\}$\;
        \For{$1 \leq i \leq 2$}{
            $MaxIC' = \{IC(T') \in IC' ~|~ \forall IC(T'') \in IC', ~ \nis{T'} \not\subseteq \nis{T''}\}$\;
            \For{$IC(T') \in MaxIC'$}{
                $CQ'_i = \{ IC(T'') \in CQ_i ~|~ \nis{T''} \subseteq \nis{T'}  \}$\;
                \If{$CQ'_i$ is a clique}{
                    $Cliques = Cliques \cup \{ CQ'_i \}$, $IC'' = IC'' \cup CQ'_i$\;
                }
            }
        }
        \If{$ IC' \subseteq IC''$}{
            break\;
        }
    }
}

$MaxCliques = \{ CQ \in Cliques ~|~ \forall CQ' \in Cliques, ~ CQ \not\subseteq CQ'\}$\;
return $MaxCliques$\;
\end{algorithm}

% In this paper, we have not developed a good algorithm to simplify \kmn SE-conditions, 
% therefore, the simplified SE-conditions are obtained by following two steps:
% \begin{itemize}
%     \item find some candidate cliques of $MGIC(k, m, n)$ automatically; 
%     \item analyze possible simplified SE-conditions manually. 
% \end{itemize}
% In what follows, we report the simplified SE-conditions of the verified \kmn problems. 

\subsection{Simplified \kmn SE-Conditions}
Now we report the simplified SE-conditions of the verified \kmn problems\footnote{
The discovered and simplified SE-conditions can be found at \url{https://github.com/wangbiu/lpmln_isets/tree/master/experimental-results/lpmln-se-conditions}.
}. 
For the 0-1-1 and 1-1-0 problems, we report the simplified SE-conditions such that I-3, I-6, and I-7 independent sets are empty. 
Let $IS_{011} = \{36, 9, 13, 18, 41, 45\}$ and $IS_{110} = IS_{011} \cup \{1, 2, 5, 33, 37\}$, 
the simplified SE-condition of the 0-1-1 problem is 
\begin{equation}
\label{eq:011-se-condition}
(I_{36} \neq \emptyset) \wedge \bigwedge_{k \in IS(0, 1, 1) - IS_{011}} (I_k = \emptyset)
\end{equation}
and the simplified SE-condition of the 1-1-0 problem is 
\begin{equation}
\label{eq:110-se-condition}
(I_{36} \neq \emptyset) \wedge \bigwedge_{k \in IS(1, 1, 0) - IS_{110}} (I_k = \emptyset)
\end{equation}

For the 0-2-1, 1-2-0, and 1-1-1 problems, we report the simplified SE-conditions  such that I-3, I-5, I-6, and I-7 independent sets are empty. 
Let $IS_{021} = \{8, 16, 64, 73, 100, 128, 146, 268, 292\}$, 
the simplified SE-condition of the 0-2-1 problem is 
\begin{equation}
\label{eq:021-se-condition}
\left(IC^1_{021} \vee IC^2_{021} \right) \wedge \expneis{I_{292}} \wedge \bigwedge_{k \in IS(0, 2, 1) - IS_{021}} \expeis{I_k} 
\end{equation}
where the formulas $IC^1_{021}$ and $IC^2_{021}$ are 
\begin{eqnarray}
& IC^1_{021} = \expeis{I_{8}} \wedge \expeis{I_{16}} \wedge \expeis{I_{268}} \\
& IC^2_{021} = \expeis{I_{64}} \wedge \expeis{I_{100}} \wedge \expeis{I_{128}}
\end{eqnarray}
Let $IS_{120} = \{1, 2, 8, 9, 10, 16, 17, 18, 73, 146, 265, 268, 289, 292\}$, 
the simplified SE-condition of the 1-2-0 problem is 
\begin{equation}
\label{eq:1-2-0-se-condition}
\left(  \expneis{I_{292}} \vee \left( \expeis{I_{292}} \wedge \expneis{I_{268}}  \wedge \expneis{I_{289}}  \right) \right)  \wedge \bigwedge_{k \in IS(1, 2, 0) - IS_{120}} \expeis{I_k}
\end{equation}

For the 1-1-1 problem, there are 19 maximal cliques of $MGIC(1, 1, 1)$. 
For brevity, we only show the simplified SE-conditions containing singleton independent sets. 
Let $IS_{111} = \{9, 18, 73,$ $ 146, 258, 265, 272, 292\}$, 
the simplified SE-conditions containing singleton independent sets are 
\begin{eqnarray}
& |I_{258}| \leq 1 \wedge \expeis{I_{272}} \wedge \expneis{I_{292}} \wedge \bigwedge_{e \in IS(1, 1, 1) - IS_{111}} \expeis{I_{e}} \\
& \expeis{I_{258}} \wedge |I_{272}| \leq 1 \wedge \expneis{I_{292}} \wedge \bigwedge_{e \in IS(1, 1, 1) - IS_{111}} \expeis{I_{e}}
\end{eqnarray}

%!TEX root = lpmln_isets.tex

\subsection{Discussion}
In this subsection, we present a discussion for discovered \kmn SE-conditions from three aspects: 
(1) the max-min \kmn non-SE-conditions; 
(2) the most general \kmn SE-conditions containing singleton independent sets (\kmn MGS-SE-conditions for short); 
and (3) the simplifications of \kmn  SE-conditions.

\textbf{About the Max-Min Non-SE-Conditions and the MGS-SE-Conditions.}
% Firstly, we discuss the results w.r.t. the max-min \kmn non-SE-conditions and the \kmn MGS-SE-conditions together. 
From the verified \kmn problems and the verified 2-1-0 SE-conditions, 
there are two interesting facts w.r.t. the max-min \kmn non-SE-conditions and the \kmn MGS-SE-conditions, 
which are shown in Fact \ref{ph:max-min-non-se-conditions} and  \ref{ph:singleton-isets-se-conditions}. 
\begin{fact}
\label{ph:max-min-non-se-conditions}
For a max-min non-SE-condition $IC(T)$ of a verified \kmn problem, 
it can be observed that $|\sis{T}| = k + m +n$.
\end{fact}

\begin{fact}
\label{ph:singleton-isets-se-conditions}
%For an MGS-SE-condition $IC(T)$ of a verified \kmn problem such that $IC(T) \in MGIC(k, m, n)$ and $|\nis{T}| > 2$ and a singleton independent set $I_i$ such that  $i \in \sis{T}$, 
%there exists an MGS-SE-condition $IC(T')$ such that $IC(T') \in MGIC(k, m, n)$, $i \in \sis{T'}$, and $|\nis{T'}| = 2$.
For an MGS-SE-condition $IC(T)$ of a verified \kmn problem that $IC(T) \in MGIC(k, m, n)$ and $|\nis{T}| > 2$ and a singleton independent set $I_i$ that  $i \in \sis{T}$, 
there exists an MGS-SE-condition $IC(T')$ that $IC(T') \in MGIC(k, m, n)$, $i \in \sis{T'}$, and $|\nis{T'}| = 2$. 
\end{fact}

Fact \ref{ph:max-min-non-se-conditions} and \ref{ph:singleton-isets-se-conditions} are important in two-fold. 
On the one hand, if these two facts hold in any \kmn problems, there is a major improvement for the searching algorithms of \kmn problems. 

By Fact \ref{ph:max-min-non-se-conditions}, the singleton non-SE-conditions of a \kmn problem can be obtained by verifying a small amount of IS-conditions. 
Specifically, the set $MNSE(k, m, n)$ can be directly computed by Equation \eqref{eq:mnse-quick}. 
\begin{equation}
\label{eq:mnse-quick}
\begin{split}
    MNSE(k, m, n) = \{ IC(T) \in NIC_x(k, m, n) ~|~ & 1 \leq x \leq k+m+n, T = \langle K, M, N \rangle, \text{ and }  \\
    & \lse{K \cup M} \neq \lse{K \cup N}
    \}
\end{split}
\end{equation}
For other IS-conditions $IC(T) \in NIC_x(k, m, n)$ such that $x > k+m+n$, 
by the NSE-preserving property of the S-EX-0 transformation, 
if there is an IS-condition $IC(T') \in MNSE(k, m, n)$ such that $\sis{T'} \subsetneq \sis{T}$, $IC(T)$ is not an SE-condition; otherwise, $IC(T)$ is an SE-condition. 
By $MGIC'(k, m, n)$, we denote the set of \kmn singleton SE-conditions $IC(T)$ such that $|\sis{T}| > k+m+n$, 
we have 
\begin{equation}
\begin{split}
MGIC'(k, m, n) = 
\{ IC(T) \in NIC_x(k, m, n) ~& |~  x > k+m+n \text{ and }\\
& \forall IC(T') \in MNSE(k, m, n), ~\sis{T'} \not\subseteq \sis{T}\}
\end{split}
\end{equation}

Furthermore, by Fact \ref{ph:singleton-isets-se-conditions}, all the singleton independent sets occurred in MGS-SE-conditions have been recognized when verifying IS-conditions $IC(T)$ such that $|\nis{T}| = 2$. 
By $SIS(k, m, n)$, we denote the set of names of such kind of independent sets, we have 
\begin{equation}
SIS(k, m, n) = \{ i ~|~ IC(T) \in NIC_2(k, m, n) \cap MGIC(k, m, n) \text{ and } i \in \sis{T} \}
\end{equation}

Combining the above facts, for a singleton IS-condition $IC(T) \in MGIC'(k, m, n)$, the MGS-SE-condition $MGIC(T')$ w.r.t. $IC(T)$ can be constructed as follows
\begin{equation}
MGIC(T') = \bigwedge_{i \in \nis{T}} \expneis{I_i} \wedge \bigwedge_{e \in \eis{T}} \expeis{I_e} \wedge \bigwedge_{s \in \nis{T} \cap SIS(k, m, n)} |I_s| = 1
\end{equation}
In other words, for \kmn problems such that $k+m+n > 2$, the IS-conditions in $MGIC'(k, m, n)$ need not be verified by Algorithm \ref{alg:compute-mg-is-condition}. 
Since Algorithm \ref{alg:compute-mg-is-condition} is the hardest part of the searching algorithms of \kmn problems, 
if Fact \ref{ph:max-min-non-se-conditions} and \ref{ph:singleton-isets-se-conditions} are true in any \kmn problems, 
the above discussion implies a major improvement for these searching algorithms.

On the other hand, Fact \ref{ph:max-min-non-se-conditions} and \ref{ph:singleton-isets-se-conditions} essentially imply the NSE-preserving property of the S-RD-1 transformation and the SE-preserving property of the S-EX-0 transformation under some unknown conditions. 
Assume under a condition $C$, the S-RD-1 transformation is NSE-preserving. 
Since the S-EX-0 transformation is the inverse of the S-RD-1 transformation, the S-EX-0 transformation is SE-preserving under the condition $C$. 
We call the properties $C$-NSE-preserving and $C$-SE-preserving properties, respectively. 

For a singleton IS-conditions $IC(T)$ in $MGIC'(k, m, n)$, 
since for any singleton IS-condition $IC(T')$ in $NIC(k, m, n)$ such that $IC(T') \subset IC(T)$, $IC(T')$ is an SE-condition, 
by the $C$-SE-preserving property of the S-EX-0 transformation, $IC(T)$ should be an SE-condition. 
Therefore, Fact \ref{ph:max-min-non-se-conditions} is derived from the $C$-SE-preserving property of the S-EX-0 transformation. 

For a \kmn MGS-SE-condition $IC(T)$ and an independent set $I_i$ such that  $i \in \sis{T}$, 
we can construct a \kmn singleton tuple $T'$ such that $IC(T') < IC(T)$, 
and let $\sex{T'} = \opsex{T, I_i, a'}$. 
It is easy to check that $T'$ is semi-strongly equivalent but $\sex{T'}$ is not. 
For a singleton IS-condition $IC(T'') \in NIC_2(k, m, n)$ such that $IC(T'') \subset IC(T')$ and $i \in \sis{T''}$, 
by the SE-preserving property of the S-RD-1 transformations, 
it is easy to check that $IC(T'')$ is an SE-condition. 
And let $\sex{T''} = \opsex{T'', I_i, a'}$, 
by the $C$-NSE-preserving property of the S-RD-1 transformation, we have $\sex{T''}$ is not semi-strongly equivalent, 
i.e., $IC(T'')$ is not an SE-condition. 
Therefore, Fact \ref{ph:singleton-isets-se-conditions} is derived from the $C$-NSE-preserving property of the S-RD-1 transformation. 

The above discussion show the $C$-NSE-preserving and $C$-SE-preserving properties may provide more understandings of the semi-strong equivalence of \lpmln and provide a major improvement for the searching algorithms of \kmn problems. 
However, it is still an open problem whether the $C$-NSE-preserving and $C$-SE-preserving properties hold. 
% which is a next work of the paper. 

\textbf{About Simplifications of the SE-conditions.}
In this paper, although we have presented a preliminary algorithm to simplify \kmn SE-conditions, 
there are still many problems on the simplifications of the SE-conditions. 

% Firstly, finding the maximal cliques of a set $MGIC(k, m, n)$ is to find a proper covering of the set, 
% i.e., different maximal cliques may have the same SE-conditions, 
% which makes the simplification complicated. 
% If different maximal cliques are not allowed to have the same SE-conditions, obtained SE-conditions may not be simplest. 
% As shown in Equation \eqref{eq:021-se-condition}, the simplified 0-2-1 SE-condition is related to two maximal cliques of the set $MGIC(0, 2, 1)$, i.e. the maximal cliques w.r.t. $IC_{021}^1$ and $IC_{021}^2$, denoted by $MC_1$ and $MC_2$.
% It is easy to check that there are 505 common empty independent sets and 1 common non-empty independent sets in both  the maximal cliques $MC_1$ and $MC_2$, 
% therefore, $MC_1$ and $MC_2$ contain $2^{511 - 506} = 2^5 = 32$ SE-conditions, respectively. 
% By Table \ref{tab:kmn-running-times}, there are 60 SE-conditions in $MGIC(0, 2, 1)$, 
% therefore, the maximal cliques  $MC_1$ and $MC_2$ have  4 common SE-conditions. 
% If maximal cliques are not allowed to have the same SE-conditions, one of the above sets $MC_1$ and $MC_2$ cannot be a maximal clique, 
% which means the simplified 0-2-1 SE-condition would be more complex. 
% Actually, there are 4 maximal cliques of $MGIC(0, 2, 1)$, if different maximal cliques do not have the same SE-conditions. 

Firstly, Algorithm \ref{alg:finding-max-cliques} may not find the simplest \kmn SE-conditions. 
Since the approach can only process IS-conditions, while the simplest SE-condition may not be an IS-condition. 
As shown in Equation \eqref{eq:lpmln-010-se-condition}, the simplest 0-1-0 SE-condition of \lpmln programs is not a disjunction of IS-conditions.
But we can only obtain the 0-1-0 SE-condition shown in Equation \eqref{eq:010-se-condition-concise} by Algorithm \ref{alg:finding-max-cliques}. 

Secondly, the optimizing approaches used in Algorithm \ref{alg:finding-se-condition-improved} make the discovered SE-conditions not easy to simplify. 
% For example, there are 40 maximal cliques in $MGIC(1, 1, 1)$. 
% But by the SE-preserving property of the S-RD-1 transformation, the simplified 1-1-1 SE-condition should be more concise. 
Specifically, we define the proper subset $MGIC_{max}(k, m, n)$ of $MGIC(k, m, n)$ as follows 
\begin{equation}
\begin{split}
MGIC_{max}(k, m, n) = \{& IC(T) \in MGIC(k, m, n) ~|~ \sis{T} = \emptyset \text{ and } \\
& \forall IC(T') \in MGIC(k, m, n), ~ \nis{T} \not\subseteq \nis{T'}\}
\end{split}
\end{equation}
By the SE-preserving property of the S-RD-1 transformation, there should be a maximal clique w.r.t. each SE-condition of $MGIC_{max}(k, m, n)$, 
i.e., for an SE-condition $IC(T) 
\in MGIC_{max}(k, m, n)$, the set $IC = \{ IC(T') \in MGIC(k, m, n) ~|~ \nis{T'} \subseteq \nis{T} \}$ should be a maximal clique. 
But since the IS-conditions containing semi-valid rules are skipped in the searching processes, 
these conditions do not occur in $MGIC(k, m, n)$, i.e., the set $IC$ is usually not a maximal clique. 

In addition, the rules in a tuple are allowed to be the same in this paper. 
Although the searching algorithms may be improved by skipping the IS-conditions containing the same rules, 
it would make the simplified SE-conditions more complex. 
For example, if we eliminate the SE-conditions containing the same rules, the simplified 0-1-1 SE-condition is 
turned into 
\begin{equation}
\label{eq:011-se-no-same-rules}
\expneis{I_{36}} \wedge \left( \expneis{I_{13}} \vee  \expneis{I_{41}} \right)  \wedge \bigwedge_{k \in IS(0, 1, 1) - IS_{011}} \expeis{I_k}
\end{equation}
Obviously, it is more complex than the 0-1-1 SE-condition in Equation \eqref{eq:011-se-condition}. 

Combining the above discussion, how to optimize the searching algorithms of \kmn problems and simplify discovered SE-conditions are still open problems, 
which would be a future work of the paper. 

%!TEX root = lpmln_isets.tex

\section{Comparison with Lin and Chen's Approach}
Lin and Chen \cite{Lin2005Discover}  have presented an approach to discovering the \kmn SE-conditions of ASP programs. 
In this section, we compare the  Lin and Chen's approach (LC-approach for short) with our independent sets approach (IS-approach for short). 
Firstly, both of the approaches are presented to discover the \kmn SE-conditions of logic programs. 
For the IS-approach, we have shown that it can be used in both ASP and \lpmlnend. 
And for the LC-approach, although it is only used in ASP, 
it is not difficult to adapt the approach to \lpmlnend, which is shown in what follows. 

Secondly, we show the differences between the LC-approach and the IS-approach.
In the LC-approach, there are four main steps to discover necessary \kmn SE-conditions:
\begin{itemize}
    \item G-step: choose a small set of atoms and generate all strongly equivalent \kmn tuples constructed from the set;
    \item C-step: conjecture a plausible \kmn syntactic condition manually; 
    \item V$_1$-step: verify the conjecture in the generated \kmn tuples automatically; 
    \item V$_2$-step: verify the conjecture in the general cases manually.
\end{itemize}
In the C-step, the LC-approach has to conjecture a syntactic condition manually, 
since there does not exist a general theorem that guarantees the forms of SE-conditions. 
In the V$_1$ step, the LC-approach can automatically verify the conjecture in the generated \kmn tuples, 
since the strong equivalence checking of ASP can be reduced to the tautology checking of a propositional formula. 
In the V$_2$ step, the LC-approach need to show the conjecture holds for arbitrary \kmn tuples, 
where the necessary part of the verification cannot be done automatically. 
A main reason is that the forms of SE-conditions are unknown. 
It is obvious that, to adapt the LC-approach to \lpmlnend, we only need to consider the V$_1$ step. 
In other words, we need to find a reduction from the semi-strong equivalence checking of \lpmln to the tautology checking of a propositional formula. 
Fortunately, the reduction has been presented in \cite{Wang2019LPMLNSEALL}, 
therefore, the LC-approach can be straightforwardly used in \lpmlnend.

In the IS-approach, there are only two steps to discover necessary SE-conditions:
\begin{itemize}
    \item G-Step: generate an IS-condition of a \kmn problem; 
    \item V-Step: verify whether the IS-condition is a \kmn SE-condition and compute the most general SE-condition. 
\end{itemize}
For a \kmn problem, there are finitely many IS-conditions and each IS-condition can be verified by comparing the HT-models of related programs, which can be done automatically. 
Since we have shown that the form of SE-conditions is exactly the form of the IS-conditions, 
the discovered SE-conditions are straightforwardly necessary.

According to the above comparison, 
the main advantage of the IS-approach is it is a fully automatic approach, 
and it can be easily adapted to other logic formalisms by studying the SE-preserving and NSE-preserving properties of S-* transformations in corresponding logic formalisms.

\section{Conclusion}
In this paper, we present a syntactic approach to study the strong equivalences of ASP and \lpmln programs. 
Firstly, we present the notions of independent set and S-* transformations  
and show the SE-preserving and NSE-preserving properties of S-* transformations. 
Secondly, we present a basic algorithm to discover \kmn SE-conditions of logic programs, which is a fully automatic approach. 
To discover the SE-conditions efficiently, we present four kinds of improvements for the algorithm. 
Due to the same properties of S-* transformations in \lpmln and ASP, the discovering approaches can be used in ASP directly. 
Thirdly, we present a maximal-cliques-based method to simplify the discovered \kmn SE-conditions and 
report the simplified 0-1-0, 0-1-1, 1-1-0, 0-2-1, 1-2-0, and 1-1-1 SE-conditions. 
Moreover, we present a discussion on two interesting facts in the discovered SE-conditions and the problems in simplifications of SE-conditions. 
Finally, we present a comparison between Lin and Chen' and our approaches to discovering SE-conditions and 
discuss the similarity and differences between the notions of S-DL, S-RD, and HT-forgetting. 
By contrast, our approach is a fully automatic approach and is easy to adapt to other logic formalisms. 
A main problem of our approach is there are too many discovered SE-conditions, 
but we have not found a good method to simplify these SE-conditions.

For the future, we will continue to study the properties of S-* transformations and IS-conditions. 
Firstly, we will investigate whether the $C$-SE-preserving property of S-EX-0 and $C$-NSE-preserving property of S-RD-1 hold in any \kmn problems of \lpmlnend. 
Secondly, we will continue to study the simplifications of \kmn SE-conditions. 
Thirdly, we will investigate why the singleton independent sets are necessary to construct an SE-condition. 

In addition, we will investigate the applications of the independent sets and S-* transformations in other aspects of logic programming. 
We believe the notions of independent sets and S-* transformations will present some new perspective for studying theoretical properties of logic programs. 
For example, the forgetting is to hide (delete) irrelevant information from a logic program, 
and HT-forgetting can preserve strong equivalence after the hiding \cite{Wang2014HTForgetting}. 
It is easy to observe that the notion of HT-forgetting is similar to the notions of S-RD and S-DL transformations, 
that is, both of the notions are used to delete information from logic programs and the deleting preserves strong equivalence.

\bibliographystyle{fundam}
\bibliography{lpmln_isets.bib}

% \input{appendix.tex}

% \newpage
% \input{s_dl_sat_table.tex}

% \newpage
% \input{s_ad_sat_table.tex}

\end{document}